\documentclass[11pt]{article}
\usepackage{amsmath,amssymb,amsfonts,amsthm,epsfig}
\usepackage[usenames,dvipsnames]{xcolor}

\usepackage{bm,xspace}

\usepackage{cancel}
\usepackage{fullpage}
\usepackage{liyang}
\usepackage{framed}
\usepackage{verbatim}
\usepackage{enumitem}
\usepackage{array}
\usepackage{multirow}
\usepackage{afterpage}
\usepackage{mathrsfs}  

\usepackage{dsfont} 
\usepackage[normalem]{ulem}

\newcommand{\lnote}[1]{\footnote{{\bf \color{blue}Li-Yang}: {#1}}}
\newcommand{\rnote}[1]{\footnote{{\bf \color{red}Rocco}: {#1}}}
\newcommand{\pparagraph}[1]{\medskip \noindent {\bf {#1}}}

\usepackage{caption} 

\usepackage{todonotes}

\makeatletter
\newtheorem*{rep@theorem}{\rep@title}
\newcommand{\newreptheorem}[2]{
\newenvironment{rep#1}[1]{
 \def\rep@title{#2 \ref{##1}}
 \begin{rep@theorem}\itshape}
 {\end{rep@theorem}}}
\makeatother
\theoremstyle{plain}

\newcommand{\ignore}[1]{}



\def\colorful{0}

\ifnum\colorful=1

\newcommand{\red}[1]{{\color{red} {#1}}}

\newcommand{\gray}[1]{{\color{gray}{#1}}}

\fi
\ifnum\colorful=0

\newcommand{\red}[1]{{{#1}}}

\newcommand{\gray}[1]{{{#1}}}

\fi

\usepackage{boxedminipage}

\newreptheorem{theorem}{Theorem}
\newtheorem*{theorem*}{Theorem}
\newreptheorem{lemma}{Lemma}
\newreptheorem{proposition}{Proposition}
\newtheorem*{noclaim*}{Claim}

\newcommand{\uhr}{\upharpoonright}

\newcommand{\encode}{\mathrm{encode}}

\newcommand{\acz}{\mathsf{AC^0}}

\newcommand{\ds}{\displaystyle}

\newcommand{\CDT}{\textsc{CDT}}
\newcommand{\CCDT}{\textsc{CCDT}}
\newcommand{\depth}{\mathrm{depth}} 

\newcommand{\SYM}{\mathsf{SYM}}

\newcommand{\simple}{\mathrm{simple}}

\newcommand{\stars}{\mathrm{stars}}
\newcommand{\gentle}{\mathrm{gentle}}

\newcommand{\PRG}{\mathrm{PRG}}
 
\newcommand{\SL}{\mathrm{SL}}

\newcommand{\Fixed}{\mathrm{Fixed}}

\renewcommand{\N}{\mathds{N}}
\renewcommand{\F}{\mathds{F}}

\begin{document}

\title{
Improved pseudorandom generators from \\ pseudorandom multi-switching lemmas
}

\author{ Rocco A.~Servedio\thanks{Supported by NSF grants CCF-1420349 and CCF-1563155. Email: {\tt rocco@cs.columbia.edu}}\\ 
Columbia University \and Li-Yang Tan\thanks{Supported by NSF grant CCF-1563122.  Part of this research was done during a visit to Columbia University. Email: {\tt liyang@cs.columbia.edu}} \\ Toyota Technological Institute}

\maketitle

\begin{abstract}

We give the best known pseudorandom generators for two touchstone classes in unconditional derandomization: an $\eps$-PRG for the class of size-$M$ depth-$d$ $\acz$ circuits with seed length $\log(M)^{d+O(1)}\cdot \log(1/\eps)$, and an $\eps$-PRG for the class of $S$-sparse $\F_2$ polynomials with seed length $2^{O(\sqrt{\log S})}\cdot \log(1/\eps)$. 
These results bring the state of the art for unconditional derandomization of these classes into sharp alignment with the state of the art for computational hardness for all parameter settings: improving on the seed lengths of either PRG would require breakthrough progress on longstanding and notorious circuit lower bounds. 

The key enabling ingredient in our approach is a new \emph{pseudorandom multi-switching lemma}.  We derandomize recently-developed \emph{multi}-switching lemmas, which are powerful generalizations of H{\aa}stad's switching lemma that deal with \emph{families} of depth-two circuits. Our pseudorandom multi-switching lemma---a randomness-efficient algorithm for sampling restrictions that simultaneously simplify all circuits in a family---achieves the parameters obtained by the (full randomness) multi-switching lemmas of Impagliazzo, Matthews, and Paturi~\cite{IMP12} and H{\aa}stad~\cite{Has14}.  This optimality of our derandomization translates into the optimality (given current circuit lower bounds) of our PRGs for $\acz$ and sparse $\F_2$ polynomials.

\end{abstract}

 \thispagestyle{empty}

\newpage

\setcounter{page}{1}

\section{Introduction} 

{\bf Switching lemmas.}  Switching lemmas, first established in a series of breakthrough works in the 1980s~\cite{ajtai1983,FSS:84,yao1985,Hastad86}, are fundamental results stating that depth-two circuits (ORs of ANDs or vice versa) simplify dramatically when they are ``hit with a random restriction.''  They are a powerful technique in circuit complexity, and are responsible for a remarkable suite of hardness results concerning small-depth Boolean circuits ($\acz$). Switching lemmas are at the heart of several near-optimal bounds on $\acz$ circuits, such as essentially optimal correlation bounds against the {\sc Parity} function~\cite{IMP12,Has14} and the worst-case and average-case depth hierarchy theorems of \cite{Hastad86,RST15,Hastad16}.  Indeed, comparably strong results are lacking (and are major open problems) for seemingly small extensions of $\acz$, such as $\acz$ augmented with parity or mod-$p$ gates, for which switching lemmas do not apply; this gap highlights the importance of switching lemmas as a proof technique.

Switching lemmas are versatile as well as powerful:  many results in circuit complexity rely on sophisticated variants and generalizations of the ``standard'' switching lemmas.  Recent examples include
the aforementioned correlation bounds and average-case depth hierarchy theorems, as well as powerful lower bounds on the circuit complexity of the {\sc Clique} problem \cite{Beame90,Rossman08}, lower bounds on the small-depth circuit complexity of {\sc st-Connectivity}~\cite{COST16}, and lower bounds against $\acz$ formulas~\cite{Rossman15}. Beyond the immediate arena of circuit lower bounds, switching lemmas are also important tools in diverse areas including propositional proof complexity \cite{PBI93,KPW94,PRST16}, computational learning theory \cite{LMN:93}, and the design of circuit satisfiability algorithms \cite{BIS12,IMP12}. 

This paper is about the role of switching lemmas in the study of \emph{unconditional pseudorandomness}.  Switching lemmas have a long history in this area; indeed, arguably the first work in unconditional derandomization, the seminal paper of Ajtai and Wigderson \cite{AjtaiWigderson:85}, was based on a \emph{pseudorandom} switching lemma, which they used to give the first non-trivial pseudorandom generator for $\acz$.  (Interestingly, after many subsequent developments described in detail in Section~\ref{sec:applic}, we come full circle in this paper and use the \cite{AjtaiWigderson:85} framework to give a new pseudorandom generator for $\acz$ that is essentially best possible without improving longstanding circuit lower bounds.)  One key contribution that we make in this paper is to bring together two important generalizations of standard switching lemmas, one quite old and one very new:

\begin{itemize}

\item [(i)] \emph{pseudorandom} switching lemmas (originating in \cite{AjtaiWigderson:85}), which employ pseudorandom rather than ``fully random'' restrictions, and 

\item [(ii)] recently developed \emph{multi-switching lemmas} \cite{IMP12,Has14} which simultaneously simplify all of the depth-two circuits in a family of such circuits, rather than a single depth-two circuit as is the case for standard switching lemmas.

\end{itemize}

Let us discuss each of these generalizations in turn.

\vspace{-5pt}
\paragraph{Pseudorandom switching lemmas.}  The (truly) random restrictions that are used in standard switching lemmas make a coordinatewise-independent random choice for each input variable $x_1,\dots,x_n$ of whether to map it to 0, to 1, or to leave it unassigned (map it to $\ast$); standard switching lemmas show that a depth-two circuit simplifies dramatically with very high probability when it is hit with such a random restriction.  Such ``truly random'' restrictions are inherently incompatible with unconditional derandomization, which naturally motivates the notion of a \emph{pseudorandom} switching lemma.  Such a result defines a much smaller probability space of ``pseudorandom'' restrictions, and proves that a restriction drawn randomly from this space also has the effect of simplifying a depth-two circuit with high probability.   While pseudorandom switching lemmas have been the subject of much research since they were first introduced by Ajtai and Wigderson \cite{AjtaiWigderson:85,Ajt93,CR96,AAIPR01,GMRTV12,IMP12,GMR13,TX13,GW14}, and have been applied in a range of different ways in unconditional derandomization, they are not yet fully understood.

The designer of a pseudorandom switching lemma faces an inherent tension between achieving strong parameters---intuitively, having a depth-two circuit simplify as much as possible while keeping a large fraction of variables alive---and using as little randomness as possible.  Prior to the work of Trevisan and Xue \cite{TX13}, known pseudorandom switching lemmas fell short of achieving the parameters of H{\aa}stad's influential ``full randomness'' switching lemma~\cite{Hastad86}.  In particular, a parameter of central importance in essentially all applications of switching lemmas is the probability that a given coordinate $x_i$ remains alive under a random (or pseudorandom) restriction; this is often referred to as the ``$\ast$-probability'' and denoted by $p$.  
A crucial quantitative advantage of H{\aa}stad's switching lemma over previous works is that it can be applied even when $p$ is as large as $\Omega(1/\log n)$ for $\poly(n)$-size depth-two circuits---in contrast, the earlier works of \cite{ajtai1983,FSS:84,yao1985} required $p=n^{-\Omega(1)}$---and yields a very strong conclusion, namely that with high probability the restricted circuit collapses to a shallow decision tree\footnote{The first published version of the switching lemma with a decision tree conclusion is due to Cai~\cite{Cai86}; several authors subsequently noted that H{\aa}stad's argument also yields such a conclusion.}.  (For example, while the recent pseudorandom switching lemma of \cite{GMR13} is able to achieve a relatively large $p$, the conclusion of that switching lemma is that the restricted depth-two circuit can w.h.p.~be sandwiched by depth-two circuits with small bottom fan-in, which is weaker than the aforementioned decision tree conclusion.)  

Trevisan and Xue \cite{TX13} give a \emph{pseudorandom} switching lemma that is highly randomness efficient and yet achieves the parameters of H{\aa}stad's fully random switching lemma (i.e.~\cite{TX13} achieves the same simplification, collapsing to a shallow decision tree, that follows from~\cite{Hastad86}, with the same $\ast$-parameter $p$ as \cite{Hastad86}).  The key conceptual ingredient enabling this is a beautiful idea of ``fooling the proof'' of the H{\aa}stad's switching lemma, exploiting its ``computational simplicity."
Trevisan and Xue leverage their pseudorandom switching lemma to construct a new pseudorandom generator for $\acz$, obtaining the first improvement of Nisan's celebrated PRG~\cite{Nis91} in over two decades. We elaborate on Trevisan and Xue's ideas and how they obtain their PRG later in Section~\ref{sec:acz}.

\vspace{-5pt}
\paragraph{Multi-switching lemmas.}  The switching lemma shows that any width-$k$ CNF formula collapses to a shallow decision tree with high probability under a random restriction.  Via a simple union bound it is of course possible to extend this result to say that a family of width-$k$ CNF formulas will all collapse to a shallow decision tree with high probability under a random restriction; but this naive approach leads to a quantitative loss in parameters if the argument is iterated, as it typically is, $d-1$ times to analyze a depth-$d$ circuit.  (The exact nature of this quantitative loss is important but somewhat subtle; see Section~\ref{sec:hastad-vs-hastad} for a detailed explanation.)

Via an ingenious extension of the ideas underlying the original switching lemma, H{\aa}stad \cite{Has14} developed ``multi-switching lemmas'' that essentially bypass this quantitative loss in parameters that results from iterating a naive union bound (see also the work of Impagliazzo, Matthews, and Paturi~\cite{IMP12}  for closely related results).  Roughly speaking,~\cite{Has14} shows that a \emph{family} of width-$k$ CNF formulas will with high probability have a shallow \emph{common partial decision tree}.  Without explaining this structure in detail here (again see Section~\ref{sec:hastad-vs-hastad} for a detailed explanation), this makes it possible to iterate the argument and tackle depth-$d$ circuits without incurring a quantitative loss in parameters.  The savings thus achieved is the key new ingredient that allowed \cite{IMP12,Has14} to achieve essentially optimal correlation bounds for $\acz$ against the  {\sc Parity}  function, capping off a long line of work~\cite{ajtai1983,yao1985,Hastad86,Cai86,babai1987,BIS12}.  These ideas have also been leveraged to achieve new algorithmic results such as better-than-brute-force satisfiability algorithms and distribution-free PAC learning algorithms for $\acz$ \cite{BIS12,IMP12,STitcs17}.

\vspace{-5pt}
\paragraph{A pseudorandom multi-switching lemma.}  A core technical contribution of this paper is to bring together these two lines of work, on pseudorandom switching lemmas and on multi-switching lemmas.  
Since the precise statement of our pseudorandom multi-switching lemma, Theorem~\ref{thm:derand-H14}, is somewhat involved we defer it to Section~\ref{sec:derand-H14} and here merely make some remarks about it. 
In the spirit of Trevisan and Xue's derandomization of the original switching lemma, to obtain Theorem~\ref{thm:derand-H14} we ``fool the proof'' of H{\aa}stad's multi-switching lemma~\cite{Has14}, exploiting its ``computational simplicity." 
 This enables us to achieve optimal parameters in the same sense as \cite{TX13}, namely, that it establishes the same dramatic simplification---now of the family $\mathscr{F}$ of depth-two circuits---as \cite{Has14}, and while only requiring the same $\ast$-probability $p$ as \cite{Has14}. Our pseudorandom switching lemma is highly efficient in its use of randomness;
this randomness efficiency is crucial in the constructions of our pseudorandom generators for $\acz$ circuits and sparse $\F_2$ polynomials using Theorem~\ref{thm:derand-H14}, which we now describe in the next section.

\section{PRGs for $\acz$ and sparse $\F_2$ polynomials} \label{sec:applic}

We employ our pseudorandom multi-switching lemma to give the best known pseudorandom generators for two canonical classes in unconditional derandomization: $\acz$ circuits and sparse $\F_2$ polynomials. 
As we describe in this section, our results bring the state of the art for unconditional derandomization of these classes into sharp alignment with the state of the art for computational hardness: improving on the seed lengths of either PRG would require breakthrough progress on longstanding and notorious circuit lower bounds.   In this sense, our results are in the same spirit as those of Imagliazzo, Meka, and Zuckerman~\cite{IMZ12}, which gave optimal (assuming current circuit lower bounds) pseudorandom generators for various classes of Boolean formulas and branching programs; however, our techniques are very different from those of~\cite{IMZ12}.

\subsection{PRGs for $\acz$ circuits} \label{sec:acz}

The class of small-depth Boolean circuits ($\acz$) is a class of central interest in unconditional derandomization, and has been the subject of  intensive research in this area over the past 30 years~\cite{AjtaiWigderson:85,LinialNisan:90,Nis91,NW94,LVW93,LV96,Kli01,Trevisan:04,Vio06,Baz09,Raz09,Bra10,KLW10,DETT10,Aaronson10,Aaronson10b,SimaZak10,LS11,FSUV12,GMRTV12,GMR13,TX13,GW14,Tal:15tightbounds,HS16}. 
This highly successful line of work on derandomizing $\acz$ has generated a wealth of ideas and techniques that have become mainstays in the field of pseudorandomness. A prominent example is Nisan's celebrated PRG for $\acz$ circuits~\cite{Nis91}, which introduced ideas that enriched the surprising connections between pseudorandomness and computational hardness~\cite{blumic82,Yao:82,NW94}.  The \emph{hardness-versus-randomness paradigm} asserts, qualitatively, that strong explicit PRGs exist if and only if strong explicit circuit lower bounds exist.  In the context of unconditional derandomization (the subject of this work), this strongly motivates the goal of constructing, for every circuit class $\mathscr{C}$, unconditional PRGs for $\mathscr{C}$ that are best possible given the current best lower bounds for $\mathscr{C}$. In other words, this is the goal of achieving a \emph{quantitatively optimal hardness to randomness conversion} for~$\mathscr{C}$, converting ``all the hardnesss" in our lower bounds for $\mathscr{C}$ into pseudorandomness for~$\mathscr{C}$.

For $\mathscr{C}$ being the class of $n$-variable size-$M$ depth-$d$ $\acz$ circuits this amounts to constructing PRGs with seed length $\log^{d-1}(Mn)\log(1/\eps)$: such seed length is best possible without improving longstanding $\acz$ lower bounds that date back to the 1980s~\cite{Hastad86}. (More precisely, it is well known, see e.g.~\cite{TX13}, that achieving seed length say $\log^{d-1.01}(Mn)\log(1/\eps)$ would yield $\exp(\omega(n^{1/(d-1)})$ size lower bounds against depth-$d$ $\acz$ circuits, which is a barrier that has stood for over 30 years even in the $d=3$ case.)  We give the first construction of a PRG that achieves this seed length up to an \emph{additive} absolute constant in the exponent of $\log(Mn)$:

\begin{theorem}[PRG for $\acz$ circuits] \label{thm:first}
For every $d\ge 2$, $M\in \N$ and $\eps > 0$, there is an $\eps$-PRG for the class of $n$-variable size-$M$ depth-$d$ circuits with seed length $\log^{d+O(1)}(Mn)\log(1/\eps)$. 
\end{theorem}

\subsubsection{Background and prior PRGs for $\acz$ circuits}  \label{sec:background-acz}

As noted above there has been a significant body of work on PRGs for $\acz$ circuits, spanning over 30 years. In this section we give a brief overview of the history and prior state-of-the-art for this touchstone problem in unconditional derandomization. 

\vspace{-5pt} 

\paragraph{Ajtai--Wigderson and Nisan.}
Ajtai and Wigderson, in their seminal work~\cite{AjtaiWigderson:85} pioneering the study of unconditional derandomization, constructed the first non-trivial PRG for $\acz$ circuits with an $n^{o(1)}$ seed length; we will discuss their techniques in detail later. \cite{AjtaiWigderson:85}'s seed length was improved significantly in the celebrated work of Nisan~\cite{Nis91}, using what is now known as the Nisan--Wigderson framework~\cite{NW94}, which provides a generic template for converting correlation bounds against a circuit class to PRGs for a closely related class (in the case of $\acz$ these two classes essentially coincide).  Via this approach Nisan showed how correlation bounds for $\acz$ against the {\sc Parity} function~\cite{Hastad86} yield a PRG with seed length  $\log^{2d+O(1)}(Mn/\eps)$. 

We remark that the generality of the Nisan--Wigderson framework comes at a quantitative price: it is straightforward to verify that a seed length of $(\log^{d}(Mn) + \log(1/\eps))^2$ is the best that can be achieved via this framework given current $\acz$ circuit lower bounds (see e.g.~\cite{TX13,HS16}). This is roughly quadratically worse than the sought-for $\log^{d-1}(Mn)\log(1/\eps)$, the best that can be achieved assuming \emph{only} current $\acz$ circuit lower bounds.

\vspace{-5pt}
\paragraph{Bounded independence fools $\acz$.} 
Nisan's seed length for $\acz$ circuits stood unmatched for more than two decades. 
However, in this interim period there was significant progress on showing that distributions with bounded independence fool $\acz$, a well-known conjecture posed by Linial and Nisan \cite{LinialNisan:90}.  Braverman's breakthrough result~\cite{Bra10} showed that $\polylog(n)$-wise independence fools $\acz$, which (along with standard constructions of $k$-wise independent distributions) gave a PRG with seed length $\log^{O(d^2)}(Mn/\eps)$; this was subsequently sharpened to $\log^{3d+O(1)}(Mn/\eps)$ by Tal~\cite{Tal:15tightbounds}.  Recently, Harsha and Srinivasan~\cite{HS16} further improved the seed length of Braverman's generator to $\log^{3d+O(1)}(Mn) \log(1/\eps)$, which is notable for its optimal dependence on the error parameter~$\eps$.

\vspace{-5pt} 
\paragraph{The work of Trevisan and Xue.} Recent work of Trevisan and Xue~\cite{TX13} makes a significant advance towards achieving seed length $\log^{d-1}(Mn)\log(1/\eps)$: their work circumvents the ``quadratic loss" associated with the Nisan--Wigderson framework with a PRG of seed length $\log^{d+O(1)}(Mn/\eps)$.  This is the first PRG to achieve a $\log^{d+O(1)}(Mn)$ dependence, an exponent that is within an \emph{additive} absolute constant of the sought-for $\log^{d-1}(Mn)$, and is also the first strict improvement on Nisan's seed length in more than two decades. (Note however, that like Nisan's PRG the dependence on $\eps$ is suboptimal: $\log^{d+O(1)}(1/\eps)$ instead of $\log(1/\eps)$.)

Rather than going through the Nisan--Wigderson framework---which, as noted above, carries with it an associated quantitative loss in parameters---Trevisan and Xue construct their PRG by \emph{derandomizing the proof} of $\acz$ lower bounds, ``opening up the black-box" of $\acz$ lower bounds, so to speak.  At a high level,~\cite{TX13} adopts the strategy employed in the early work of Ajtai and Wigderson~\cite{AjtaiWigderson:85}. We describe this strategy in detail in Section~\ref{sec:AW}, but roughly speaking, Ajtai and Wigderson introduced a powerful and generic framework for constructing PRGs from pseudorandom switching lemmas. In~\cite{AjtaiWigderson:85}, they instantiated this framework with a derandomization of Ajtai's switching lemma~\cite{ajtai1983}---which underlies his proof of the first superpolynomial lower bounds against $\acz$---to obtain the first non-trivial PRG for $\acz$.  Trevisan and Xue obtain their PRG by revisiting this early framework of~\cite{AjtaiWigderson:85}, instantiating it with their derandomization of H{\aa}stad's switching lemma~\cite{Hastad86}. (And as we will soon discuss, in this work we obtain our PRG by instantiating the~\cite{AjtaiWigderson:85} framework with our derandomization of the~\cite{Has14} multi-switching lemmas.)

\subsubsection{Our PRG and approach} 
To summarize, prior to our work there were two incomparable best known PRGs for $\acz$: the PRG of Trevisan and Xue~\cite{TX13}, which has seed length $\log^{d+O(1)}(Mn/\eps)$, and Harsha and Srinivasan's improvement of Braverman's generator~\cite{HS16}, which has seed length $\log^{3d+O(1)}(Mn)\log(1/\eps)$.

Theorem~\ref{thm:first} unifies and improves these incomparable seed lengths.  Our PRG  achieves an essentially optimal hardness to randomness conversion for $\acz$: our seed length of $\log^{d+O(1)}(Mn) \log(1/\eps)$ comes very close to $\log^{d-1}(Mn)\log(1/\eps)$, which is best possible without improving longstanding $\acz$ circuit lower bounds that date back to the 1980s.  (We reiterate that any PRG obtained within the Nisan--Wigderson framework must have seed length at least $(\log^{d}(Mn) + \log(1/\eps))^2$ given the current state of circuit lower bounds.)

Table~\ref{table:ac0} provides a comparison of the seed length of our PRG (and the techniques that underlie our construction) and those of previous work.

\begin{table}[h]
\renewcommand{\arraystretch}{1.6}
\centering
\begin{tabular}{|m{7em}|m{11em}|m{21em}|}
\hline
 Reference &  Seed length & Techniques \\ \hline
\cite{AjtaiWigderson:85} & $n^{o(1)}$ for $M=\poly(n)$  & derandomize~\cite{ajtai1983} switching lemma \\ \hline
\cite{Nis91} &  $\log^{2d+O(1)}(Mn/\eps)$ & \cite{NW94} framework, \cite{Hastad86} correlation bounds \\ \hline
\cite{Bra10} & $\log^{O(d^2)}(Mn/\eps)$ & bounded independence\\ \hline 
\cite{TX13} & $\log^{d+O(1)}(Mn/\eps)$ & \cite{AjtaiWigderson:85} framework, derandomize~\cite{Hastad86} switching lemma\\ \hline  
\cite{Tal:15tightbounds} & $\log^{3d+O(1)}(Mn/\eps)$ & bounded independence \\ \hline
\cite{HS16} & $\log^{3d+O(1)}(Mn)  \log(1/\eps)$ & bounded independence \\ \hline \hline
{\bf This work} & $\log^{d+O(1)}(Mn) \log(1/\eps)$ & \cite{AjtaiWigderson:85} framework, derandomize~\cite{Has14} multi-switching lemma,
bounded independence\\ \hline
 \end{tabular}
 \caption{PRGs for $\eps$-fooling $n$-variable size-$M$ depth-$d$ $\acz$ circuits.}
 \label{table:ac0}
 \end{table}


\pparagraph{Our approach.}  Our approach draws on and unifies ideas in the works of~\cite{AjtaiWigderson:85,TX13,HS16} discussed above, which we use in conjunction with our derandomization of the~\cite{Has14} multi-switching lemma to obtain our PRG.  

At a high level, we adopt the overall conceptual strategy of Ajtai and Wigderson~\cite{AjtaiWigderson:85} and Trevisan and Xue~\cite{TX13}, and obtain our PRG by derandomizing the proof of $\acz$ lower bounds.  The key technical ingredient in our PRG construction is our pseudorandom multi-switching lemma, a derandomization of the multi-switching lemmas which underlie the~\cite{IMP12,Has14} optimal correlation bounds for $\acz$ against {\sc Parity}.  Our pseudorandom multi-switching lemma improves both the pseudorandom switching lemma of~\cite{TX13} (a derandomization of H{\aa}stad's switching lemma~\cite{Hastad86} which underlies his exponential lower bounds against $\acz$) and the pseudorandom switching lemma of~\cite{AjtaiWigderson:85} (a derandomization of Ajtai's switching lemma~\cite{ajtai1983} which underlies his superpolynomial lower bounds against $\acz$).

Our derandomization of the~\cite{Has14} multi-switching lemma is largely influenced by Trevisan and Xue's derandomization of the H{\aa}stad's original switching lemma~\cite{Hastad86}.  We describe our approach in detail in Section~\ref{sec:derand-H14}, but highlight here the simple but ingenious new idea underlying~\cite{TX13}'s argument. Very roughly speaking, they derandomize the~\cite{Hastad86} switching lemma by ``fooling its proof'':  showing that H{\aa}stad's proof of his switching lemma ``cannot $\delta$-distinguish" between truly random restrictions and pseudorandom restrictions drawn from $\polylog(n)$-wise independent distributions.  Since H{\aa}stad's switching lemma holds for truly random restrictions, it thus follows that it also holds for pseudorandom restrictions drawn from $\polylog(n)$-wise independent distributions (up to a $\delta$ additive loss in the failure probability). 

To accomplish this, Trevisan and Xue exploit the fact that H{\aa}stad's proof of the switching lemma is ``computationally simple": for a fixed $k$-CNF $F$, there is a small depth-$3$ circuit that takes as input an encoding of a restriction $\rho$, and outputs $1$ iff $\rho$ is a bad restriction for the desired conclusion of H{\aa}stad's switching lemma, contributing to its failure probability (more precisely, the failure event is that the ``canonical decision tree" for $F \uhr \rho$ has large depth).  In similar spirit, our derandomization of the~\cite{Has14} multi-switching lemma also exploits the ``computational simplicity" of their proofs. In our case, for a fixed family $\mathscr{F}$ of $k$-CNF formulas we construct a small depth-$4$ circuit for recognizing bad restrictions (the one additional layer of depth reflects the fact that multi-switching lemmas are, roughly speaking, ``one quantifier more complex" than switching lemmas). To obtain optimal parameters in our PRG constructions, we use the $d=3$ case of Harsha and Srinivasan's strengthening of Braverman's generator~\cite{HS16} to fool this depth-$4$ circuit, and hence show that~\cite{Has14}'s proofs of the multi-switching lemmas ``cannot distinguish" between truly random and pseudorandom restrictions. The fact that~\cite{HS16} achieves an optimal $\log(1/\eps)$ seed length dependence plays a crucial role in the optimal $\log(1/\eps)$ seed length dependence of our PRG.

\subsection{PRGs for sparse $\F_2$ polynomials} \label{sec:F2}

Our second main result deals with the class of sparse $\F_2$ polynomials.  Like $\acz$ circuits, sparse $\F_2$ polynomials and low-degree $\F_2$ polynomials have been extensively studied in unconditional derandomization~\cite{NN93,AGHP92,LVW93,Bog05,Vio06,Lov09,Vio09,BV10,LS11,Lu12}.   

Via the hardness-versus-randomness paradigm, the problem of derandomizing $\F_2$ polynomials is intimately related to that of proving correlation bounds for $\F_2$ polynomials.  A prominent open problem in the latter context---arguably the current flagship challenge in this area---is that of obtaining superpolynomially small correlation bounds against $\F_2$ polynomials of degree $\log n$.  Degree $\log n$ represents the fundamental limit of our current suite of powerful techniques for proving $\F_2$ correlation bounds~\cite{BNS92,Bou05,Cha07,VW07}, and breaking this ``degree $\log n$ barrier" would constitute a significant technical breakthrough\footnote{Breaking this ``degree $\log n$ barrier" is also well-known (via a simple and beautiful observation of H{\aa}stad and Goldmann~\cite{HG91}) to be a prerequisite for breaking the notorious ``$\log n$  party barrier" in multi-party communication complexity~\cite{BNS92}, a longstanding open problem that has resisted attack for over two decades.}.  See Open Question 1 of Viola's excellent survey~\cite{Viola09now} for a detailed discussion of this important open problem and its relationship with other central challenges in complexity theory.

As a second application of our pseudorandom multi-switching lemma, we give an $\eps$-PRG for $S$-sparse $\F_2$ polynomials with seed length $2^{O(\sqrt{\log S})}\log(1/\eps)$, which is best possible without breaking the aforementioned ``degree $\log n$ barrier" for $\F_2$ correlation bounds: 


\begin{theorem}[PRG for sparse $\F_2$ polynomials] \label{thm:second}
For every $S =  2^{\omega(\log \log n)^2}$ and $\eps > 0$ there is a PRG with seed length $2^{O(\sqrt{\log S})}\log(1/\eps)$ that $\eps$-fools the class of $n$-variable $S$-sparse $\F_2$ polynomials. 
\end{theorem}

\paragraph{Background and prior PRGs for $\F_2$ polynomials.}

The first unconditional PRGs for $\F_2$ polynomials were given in early influential work of Luby, Veli{\v{c}}kovi{\'c}, and Wigderson~\cite{LVW93}, who constructed a PRG that $\eps$-fools size-$S$ $\SYM \circ \AND$ circuits---including $S$-sparse $\F_2$ polynomials as an important special case---with seed length $2^{O(\sqrt{\log (S/\eps)})}$.  To obtain their PRG, Luby et al.~employed the Nisan--Wigderson framework~\cite{NW94} together with multi-party number-on-the-forehead (NOF) communication complexity lower bounds from the seminal work of Babai, Nisan, and Szegedy~\cite{BNS92}. Viola \cite{Vio06} subsequently extended this $2^{O(\sqrt{\log (S/\eps)})}$ seed length to the broader class of $\SYM \circ \acz$ circuits with a more modular proof.  

In a related line of work, PRGs for \emph{low-degree} $\F_2$ polynomials have also been intensively studied.  Starting with the fundamental results of Naor and Naor~\cite{NN93} on $\eps$-biased distributions (which resolved the degree-$1$ case), this research continued through an exciting line of work on the degree $k \ge 2$ case \cite{Bog05,BV10} and culminated in the breakthroughs of Lovett~\cite{Lov09} and Viola~\cite{Vio09} which are described in more detail below.\ignore{ for fixed $d$, give seed length $O(\log n + \log(1/\eps)).$}  It is interesting to note that prior to our work, the underlying techniques used for the sparse case (multi-party communication complexity) are completely different from the techniques used for the low-degree case (Fourier analysis).

\vspace{-5pt} 
\paragraph{Our PRG and approach.} Theorem~\ref{thm:second} gives an exponential and optimal improvement of the PRG of~\cite{LVW93} in terms of its dependence on the error parameter $\eps$. Our PRG  achieves an optimal hardness to randomness conversion for $\F_2$ polynomials: since every $\log(n)$-degree $\F_2$ polynomial has at most $n^{\log n}$ monomials, it can be shown (using the simple Proposition 3.1 of \cite{Vio09}) that a PRG with seed length $2^{o(\sqrt{\log S})}\log(1/\eps)$ would break the degree $\log n$ barrier. (Similar to the situation for $\acz$ circuits, it is straightforward to verify that our optimal $\log(1/\eps)$ dependence is not achievable via the Nisan--Wigderson framework without dramatic breakthroughs in correlation bounds for $\F_2$ polynomials, going well beyond breaking the degree $\log n$ barrier.) 

Our approach to obtaining Theorem~\ref{thm:second} bridges the two previously disparate lines of work on pseudorandomness for sparse and low degree polynomials: roughly speaking, it can be viewed as a reduction from PRGs for $S$-sparse polynomials to PRGs for degree-$\sqrt{\log S}$ polynomials. This allows us to leverage the result of Viola~\cite{Vio09} (building on the work of Lovett~\cite{Lov09}), which gives PRGs for $n$-variable degree-$k$ $\F_2$ polynomials with seed length 
\[ O(k\log n + k2^{k}\log(1/\eps)). \] 
More precisely, at the heart of our reduction is a new pseudorandom switching lemma for sparse $\F_2$ polynomials, showing that such a polynomial is very likely to collapse to a \emph{small-depth decision tree with low-degree $\F_2$ polynomials at its leaves} under a suitable pseudorandom restriction.  This is essentially a special case of our pseudorandom multi-switching lemma.  With this reduction in hand, we then exploit the strength and generality of Viola's result---roughly speaking, that the sum of $k$ independent copies of a sufficiently strong $\eps$-biased distribution fools degree-$k$ polynomials---to show that his PRG extends to fool not only low-degree polynomials, but also small-depth decision trees with low-degree polynomials at their leaves. 

Table~\ref{table:F2} provides a comparison of the seed length of our PRG (and the techniques that underlie our construction) and those of previous work.

\begin{table}[h]
\renewcommand{\arraystretch}{1.6}
\centering
\begin{tabular}{|m{6em}|m{11.5em}|m{21em}|}
\hline
 Reference/ Class&  Seed length & Techniques \\ \hline
\cite{LVW93} $\text{~~~~~}$ $S$ sparse & $2^{O(\sqrt{\log(S/\eps)})}$ & \cite{NW94} framework, \cite{BNS92} multi-party NOF communication complexity\\ \hline
\cite{Lov09} $\text{~~~~~~~}$ degree $k$ & $O(2^k \log n + 4^{k} \log(1/\eps))$ & Fourier analysis \\ \hline 
\cite{Vio09} $\text{~~~~~}$ degree $k$ & $O(k\log n + k2^{k}\log(1/\eps))$  & Fourier analysis \\ \hline  \hline
{\bf This work} & $2^{O(\sqrt{\log S})}  \log(1/\eps)$ & \cite{AjtaiWigderson:85} framework, derandomize~\cite{Has14} multi-switching lemma,
Fourier analysis, bounded independence\\ \hline
 \end{tabular}
 \caption{PRGs for $\eps$-fooling $\F_2$ polynomials.}
 \label{table:F2}
 \end{table}



\subsection{Organization}
Section \ref{sec:prelim} recalls some basic preliminaries from unconditional pseudorandomness.
We describe and contrast the original H{\aa}stad switching lemma~\cite{Hastad86} versus the \cite{Has14} multi-switching lemma in Section~\ref{sec:hastad-vs-hastad}.
Section~\ref{sec:ccdt} establishes some infrastructure towards derandomizing the \cite{Has14} switching lemma, and the actual derandomization is carried out in Section~\ref{sec:derand-H14}, culminating in the proof of Theorem~\ref{thm:derand-H14}.
Section~\ref{sec:AW} describes a general framework for constructing pseudorandom generators that is implicit in the work of Ajtai and Wigderson~\cite{AjtaiWigderson:85}; a crucial ingredient in this framework for constructing a pseudorandom generator for a class $\mathscr{C}$ is a ``pseudorandom simplification lemma'' for $\mathscr{C}.$  In Section~\ref{sec:prsl} we apply our derandomized multi-switching lemma from Section~\ref{sec:derand-H14} to obtain the required pseudorandom simplification lemmas for $\acz$ circuits and for sparse $\F_2$ polynomials.  Finally, Section~\ref{sec:puttogether} puts the pieces together and establishes the PRGs for $\acz$ and for sparse $\F_2$ polynomials that are our main PRG results.

\subsection{Preliminaries} \label{sec:prelim}


For $r < n$, we say that a distribution $\calD$ over $\zo^n$ can be \emph{sampled efficiently with $r$ random bits} if (i) $\calD$ is the uniform distribution over a multiset $z^{(1)},\dots,z^{(s)}$ of strings from $\zo^n$  where $s \in [{\frac 1 {\poly(n)}} \cdot 2^r, 2^r]$ and (ii) there is a deterministic algorithm $\mathrm{Gen}_{\calD}$ which, given as input a uniform random element of $[s]$, runs in time $\poly(n,s)$ and outputs a string drawn from $\calD$.

%
%

For $\delta>0$ and a class $\mathscr{C}$ of functions from $\zo^n$ to $\zo$, we say that a distribution $\calD$ over $\zo^n$ \emph{$\delta$-fools $\mathscr{C}$ with seed length $r$} if (a) $\calD$ can be sampled efficiently with $r$ random bits via algorithm $\mathrm{Gen}_\calD$, and (b) for every function $f \in \mathscr{C}$, we have
\[
\bigg|\Ex_{\bs \leftarrow \{0,1\}^r}\big[f(\mathrm{Gen}_{\calD}(\bs))\big] - 
\Ex_{\bx \leftarrow \{0,1\}^n}\big[f(\bx)\big]\bigg| \leq \delta.
\]
Equivalently, we say that $\mathrm{Gen}_\calD$ is a \emph{$\delta$-PRG for $\mathscr{C}$ with seed length $r$.}

Two kinds of distributions which are extremely useful in derandomization are \emph{$\delta$-biased} and \emph{$k$-wise independent} distributions.
We say that a distribution $\calD$ over $\zo^n$ is \emph{$\delta$-biased} if it $\delta$-fools the class 
of all $2^n$ parity functions $\{\text{{\sc Parity}}_S\}_{S \subseteq [n]}$, where $\text{{\sc Parity}}_S: \zo^n \to \zo$ is defined
by $\text{{\sc Parity}}_S(x) = \sum_{i \in S} x_i \mod 2.$
We say that a distribution $\calD$ over $\zo^n$ is \emph{$k$-wise independent with parameter $p$} if for every $1 \leq i_1 < \cdots < i_k \leq n$ and every $(b_1,\dots,b_k) \in \zo^k$, we have
\[
\Prx_{\bx \leftarrow \calD}\big[\bx_{i_1} = b_1 \text{~and~} \cdots \text{~and~}\bx_{i_k} = b_k\big] = p^{\sum_{j=1}^k b_j}  \cdot
(1-p)^{k - \sum_{j=1}^k b_j},
\]
i.e. every subset of $k$ coordinates is distributed identically to a product distribution with parameter~$p$.

A \emph{restriction} $\rho$ of variables $x_1,\dots,x_n$ is an element of $\{0,1,\ast\}^n$. We write $\supp(\rho)$ to denote the set of coordinates that are fixed to 0 or 1 by $\rho$.
Given a function $f(x_1,\dots,x_n)$ and a restriction $\rho$, we write $f \uhr \rho$ to denote the function obtained by fixing $x_i$ to $\rho(i)$ if $\rho(i) \in \{0,1\}$ and leaving $x_i$ unset if $\rho(i)=\ast.$
 For two restrictions $\rho,\rho'\in \{0,1,\ast\}^{n}$, their \emph{composition}, denoted $\rho\rho'\in \{0,1,\ast\}^{n}$, is the restriction defined by
\[ (\rho\rho')_i = \left\{
\begin{array}{cl}
\rho_i & \text{if $\rho_i \in \{0,1\}$} \\
\rho'_i & \text{otherwise.}
\end{array}
\right.\]
Given a collection $\mathscr{F} = \{f_1,\ldots,f_M\}$ of functions and a restriction $\rho$ we write $\mathscr{F} \uhr \rho$ to denote the family $\{f_1 \uhr \rho,\ldots,f_M\uhr\rho\}$. 

Given an $\acz$ circuit, we define its size to include the input variables (along with the number of gates in the circuit). We adopt this convention for notational convenience, since we may then always assume that the size $M$ of an $n$-variable circuit is always at least $n$. (We do \emph{not} adopt this convention for $\F_2$ polynomials: as is standard, we define the sparsity of an $\F_2$ polynomial to be the number of monomials in its support.)

Finally, if $g$ is a Boolean function and $\mathscr{C}$ is a class of circuits, we say that
$g$ is \emph{computed by a $(t,\mathscr{C})$-decision tree} if $g$ is computed by a decision tree of depth $t$ (with single Boolean 
variables $x_i$ at internal nodes as usual) in which each leaf is labeled by a function from $\mathscr{C}.$


\section{Multi-switching lemmas} \label{sec:hastad-vs-hastad}

At the heart of almost all applications of H{\aa}stad's original switching lemma~\cite{Hastad86} is a powerful structural fact about $\acz$ circuits: every $\acz$ circuit ``collapses" (i.e.~simplifies dramatically) to a depth-$t$ decision tree with high probability, at least $1-\eps$, under a random restriction that randomly fixes a $(1-p)$-fraction of coordinates.  In the precise quantitative statement of this fact, both $t$ and $p$ depend on~$\eps$: as the desired failure probability $\eps$ tends to $0$, the $\ast$-probability $p$ tends to $0$ (more coordinates are fixed) and $t$ tends to $n$ (the resulting decision tree is of larger depth). It is easy to see that this dependence is inherent given the statement of the~\cite{Hastad86} switching lemma, and indeed this will be clear from the discussion later in this section.

The recent multi-switching lemma of H{\aa}stad~\cite{Has14} (see also~\cite{IMP12}) achieves a remarkable strengthening of the above: essentially the same structural fact about $\acz$ holds (in terms of the quantitative relation between the decision tree depth $t$ and the failure probability $\eps$) \emph{with the $\ast$-probability $p$ being independent of $\eps$}. This is the key qualitative difference underlying the optimal $\acz$ correlation bounds for {\sc Parity} obtained in~\cite{IMP12,Has14}; likewise, in this work, this is the key qualitative difference underlying the optimal $\eps$-dependence in the seed lengths of our PRGs for $\acz$ circuits and sparse $\F_2$ polynomials.\medskip

Let $\calR_p$ denote the random restriction which independently sets each variable $x_i$ to $0$ with probability
$(1-p)/2$, to 1 with probability $(1-p)/2$, and to $\ast$ with probability $p.$  We first recall the original switching lemma
from \cite{Hastad86}:

\begin{theorem}[H{\aa}stad's switching lemma]
\label{thm:HSL}
Let $F$ be a $k$-CNF. Then for all $t \ge 1$, we have that 
\[ \Prx_{\brho\leftarrow\calR_p}[F\uhr\brho \text{~does not have a decision tree of depth~}t\, ] \le (5pk)^t. \] 
\end{theorem} 

In the context of $\acz$ circuits the switching lemma is used to achieve \emph{depth reduction} under random restrictions: we apply Theorem~\ref{thm:HSL} separately to each of the bottom-layer depth-$2$ subcircuits, choosing $t$ appropriately so that all of them ``switch" to depth-$t$ decision trees with high probability.  The following corollary is what is typically used:

\begin{corollary}[$\acz$ depth reduction via Theorem~\ref{thm:HSL}]
\label{cor:HSL-depth-reduction} 
Let $\calC$ be a size-$M$ depth-$d$ $\acz$ circuit with bottom fan-in~$k$, and let $p = 1/(10k)$.  Then for all $\eps > 0$,
\[ \Prx_{\brho\leftarrow\calR_p}\big[\, \text{$\calC \uhr \brho$ is not computed by a depth-$(d-1)$ circuit with bottom fan-in $\log(M/\eps)$}\big] \le \eps.\] 
\end{corollary}

\begin{proof}
This follows from applying Theorem~\ref{thm:HSL} with $t = \log(M/\eps)$ to each of the bottom-layer depth-$2$ subcircuits of $\calC$ (at most $M$ of them), along with the basic fact that a depth-$t$ decision tree can be expressed as both a $t$-DNF as well as a $t$-CNF. 
\end{proof} 

\ignore{
\gray{

\begin{corollary}[Theorem~\ref{thm:HSL} + union bound]
\label{cor:HSL}
Let $\mathscr{F} = \{F_1,\ldots,F_M\}$ be a collection of $k$-CNFs and $p := 1/(10k)$. Then 
\[ \Prx_{\brho\leftarrow\calR_p}[\, \text{some $F_i \in \mathscr{F}$ does not have a decision tree of depth~}t := \log(M/\eps) \,] \le \eps. \] 
\end{corollary} 

This corollary enables one to collapse the depth of an $\acz$ circuit by one, by replacing the bottom three levels of
the circuit (say OR-AND-OR) with OR of depth-$t$ decision trees, which is equivalent to a $t$-DNF.}  
}

The same argument is then repeated again on the $(k= \log(M/\eps))$-DNFs at the bottom two layers of the new circuit (applying the dual form of the switching lemma for $k$-DNFs rather than $k$-CNFs) to further reduce the depth to $d-2$.  However, observe that in this second application of the switching lemma (and in later applications as well), in order to use Corollary \ref{cor:HSL-depth-reduction}, the parameter $p$ of the random restriction must now depend on $\eps$, since we must now take $p < 1/({5k}) =  1/(5\log(M/\eps))$ in order to get a nontrivial bound in Theorem~\ref{thm:HSL}. This is why standard applications of the~\cite{Hastad86} switching lemma (involving $d-1$ iterative applications of Corollary~\ref{cor:HSL-depth-reduction}) show that every size-$M$ depth-$d$ $\acz$ circuit collapses to depth-$(t = \log(M/\eps))$ decision tree with high probability, at least $1-\eps$, under a random restriction with $\ast$-probability $p = \Theta(1/\log^{d-1}(M/\eps))$. Note that $t$ and $p$ both depend on~$\eps$.

As alluded to above, the recent multi-switching lemma of~\cite{Has14} shows, remarkably, that essentially the same simplification holds under a random restriction with $\ast$-probability $p = \Theta(1/\log^{d-1}(M))$, independent of $\eps$. Let us establish some terminology and notation to present these results.

\begin{definition}[Common partial decision tree] 
Let $\mathscr{F} = \{F_1,\ldots,F_M\}$ be a collection of Boolean functions. We say that a decision tree $T$ is a \emph{common $\ell$-partial decision tree for $\mathscr{F}$} if every $F_i \in \mathscr{F}$ can be expressed as $T$ with depth-$\ell$ decision trees at its its leaves. (Equivalently, for every $F_i \in \mathscr{F}$ and root-to-leaf path $\pi$ in $T$, we have that $F_i \uhr \pi$ is computed by a depth-$\ell$ decision tree.)  
\end{definition}

The multi-switching lemma of \cite{Has14} is as follows:

\begin{theorem}[Multi-switching lemma, Lemma 3.8 of~\cite{Has14}]
\label{thm:HSL14}
Let $\mathscr{F} = \{F_1,\ldots,F_M\}$ be a collection of $k$-CNFs and $\ell := \log(2M)$. Then for all $t \ge 1$, 
\[ \Prx_{\brho\leftarrow\calR_p}[\, \mathscr{F}\uhr\brho \text{~does not have a common $\ell := \log(2M)$-partial DT of depth $t$}\,] \le M(24pk)^t.  \] 
\end{theorem}

The following corollary should be contrasted with Corollary~\ref{cor:HSL-depth-reduction}:

\begin{corollary}[$\acz$ depth reduction via Theorem~\ref{thm:HSL14}; c.f.~Corollary~\ref{cor:HSL-depth-reduction}] 
\label{cor:HSL14-depth-reduction}
Let $\calC$ be a size-$M$ depth-$d$ $\acz$ circuit with bottom fan-in~$k$, and let $p = 1/(48k)$.  Then for all $\eps > 0$,
\[ \Prx_{\brho\leftarrow\calR_p}\big[\, \text{$\calC \uhr \brho$ is not computed by a $((\log(M/\eps),\acz(\text{depth $d-1$, bottom fan-in $\log (2M)$})$-decision tree})\,\big] \le \eps.\] 
\end{corollary}

\begin{proof}
This follows by applying Theorem~\ref{thm:HSL14} with $\mathscr{F}$ being the bottom-layer depth-$2$ subcircuits of $\calC$ and $t = \log(M/\eps)$, along with the fact that a depth-$\ell$ decision tree can be expressed as both a $\ell$-DNF and an $\ell$-CNF. 
\end{proof} 

\ignore{
\gray{
\begin{corollary}[cf.~Corollary~\ref{cor:HSL}]
\label{cor:HSL14-cor} 
Let $\mathscr{F} = \{F_1,\ldots,F_M\}$ be a collection of $k$-CNFs, $p := 1/(48k)$ and $\ell := \log(2M)$. Then 
\[ \Prx_{\brho\leftarrow\calR_p}[\, \mathscr{F}\uhr\brho\text{~does not have a common $\ell := \log(2M)$-partial DT of depth $t := \log(M/\eps)$}\,] \le \eps. \] 
\end{corollary} 
}
}

We highlight a crucial qualitative aspect of Corollary~\ref{cor:HSL14-depth-reduction}: while the depth $t=\log(M/\eps)$ of the decision tree whose existence it asserts does depend on $\eps$, the depth-$(d-1)$ $\acz$ circuits at its leaves have bottom fan-in $k = \log(2M)$ which does \emph{not} depend on $\eps$.  This means that in successive application of Corollary~\ref{cor:HSL14-depth-reduction}, the values of $p = 1/(48k) = \Theta(1/\log M)$ will remain independent of $\eps$. This leads to much better quantitative bounds than can be obtained through repeated applications of Corollary~\ref{cor:HSL-depth-reduction}: $d-1$ iterative applications of Corollary~\ref{cor:HSL14-depth-reduction} imply  that every size-$M$ depth-$d$ $\acz$ circuit collapses to a depth-$O(2^d\log(M/\eps))$ decision tree with high probability, at least $1-\eps$, under a random restriction with $\ast$-probability $p = \Theta(1/\log^{d-1} M)$. Note that the overall $\ast$-probability $p$ is independent of $\eps$.


\vspace{-5pt} 
\paragraph{Multi-switching lemmas and sparse $\F_2$ polynomials.} The qualitative advantage of multi-switching lemmas---in particular, the crucial role of a common partial decision tree---can also be seen within the context of $\F_2$ polynomials.  

Let $P$ be an $S$-sparse $\F_2$ polynomial. It is an easy observation that $P$ becomes a low-degree polynomial with high probability when hit with a random restriction: for all $\eps, p\in (0,1)$ and $k\in \N$, 
\begin{equation} \Prx_{\brho\leftarrow\calR_{\frac{p}{2}}}[\, P\uhr \brho \text{~is not a degree-$k$ polynomial}\,] \le \frac{\eps}{2} + S{w\choose k}p^k \quad \text{where $w = \Theta(\log(S/\eps))$.} \label{eq:naive}
\end{equation} 
(The proof follows by considering each monomial of $P$ individually and taking a union bound over all $S$ of them. For a fixed monomial, the probability that more than $\Omega(\log(S/\eps))$ variables survive a random restriction from $\calR_{\frac{1}{2}}$ is at most $\eps/(2S)$; next, the probability that at least $k$ variables in a width-$w$ monomial survive a random restriction from $\calR_p$ is at most ${w \choose k}p^k$.)   The failure probability of (\ref{eq:naive}) can be made at most $\eps$ by choosing $p$ and $k$ appropriately, but note that at least one of $p$ (the $\ast$-probability) or $k$ (the degree of the resulting polynomial) must depend on~$\eps$. 

Using a slight extension of the ideas in the multi-switching lemmas of~\cite{Has14}, we can instead bound the probability that $P \uhr \brho$ becomes a \emph{depth-$t$ decision tree with degree-$k$ polynomials at its leaves}.  While this provides weaker structural information than the simple observation above (cf.~Corollary~\ref{cor:HSL-depth-reduction} vs.~Corollary~\ref{cor:HSL14-depth-reduction} in the context of $\acz$), the crucial win will come from the fact that $p$ and $k$ can \emph{both} be taken to be independent of the failure probability $\eps$ (and only $t$ will depend on~$\eps$).

\subsection{Canonical common $\ell$-partial decision trees} \label{sec:ccdt}

An important concept in the proof of Theorem~\ref{thm:HSL14} is that of a \emph{canonical} common $\ell$-partial decision tree for an ordered collection $\mathscr{F}$ of $k$-CNFs, which we define in this section.

Given a $k$-CNF formula $F$ (which we view as an ordered sequence of width-$k$ clauses $C_1 \wedge C_2 \wedge \cdots$), we recall the notion of the \emph{canonical decision tree} for $F$, denoted $\CDT(F)$.  This is a decision tree which computes $F$ and is obtained as follows:  

\begin{itemize}

\item If any clause $C_i$ is identically-0, then the tree is the constant 0.  

\item If every clause $C_i$ is identically-1, then the tree is the constant 1.

\item Otherwise, let $C_{i_1}$ be the first clause that is not identically-1, and let $\kappa \in [k]$ be the number of variables in $C_{i_1}$. The first $\kappa$ levels of $\CDT(F)$ exhaustively query these $\kappa$ variables.  At each of the $2^{\kappa}$ resulting leaves of the tree (each one corresponding to some restriction ${\eta} \in \zo^{\kappa}$ fixing those $\kappa$ variables), recursively put down the canonical decision tree $\CDT(F \uhr {\eta}).$

\end{itemize}
We observe that the tree $\CDT(F)$ is unique given a fixed ordering $C_1,C_2,\ldots$ of the clauses in $F$.

H{\aa}stad's proof of his original switching lemma (Theorem~\ref{thm:HSL}) actually shows that if $F$ is a $k$-CNF, then the canonical decision tree $\CDT(F \uhr \brho)$ is shallow w.h.p.~over $\brho \leftarrow \calR_p$.  This is crucially important for the arguments of Trevisan and Xue~\cite{TX13}, who give a \emph{derandomized} version of H{\aa}stad's original switching lemma:  they construct a pseudorandom distribution over restrictions to take the place of $\calR_p$, and show that with high probability a restriction drawn from this pseudorandom distribution causes a $k$-CNF to collapse to a small-depth decision tree.  Their argument uses the structure of a canonical decision tree in an essential way.

Turning to H{\aa}stad's multi-switching lemma~\cite{Has14}, we observe that analogous to his original switching lemma, the proof of Theorem~\ref{thm:HSL14} given in \cite{Has14} implicitly establishes a stronger statement: $\mathscr{F}\uhr\brho$ has a small-depth \emph{canonical} common $\ell$-partial decision tree w.h.p.~over $\brho\leftarrow\calR_p$.  In fact, we will use the fact that it actually establishes an even stronger statement: w.h.p.~over $\brho\leftarrow\calR_p$,  \emph{every} canonical common $\ell$-partial decision tree for $\mathscr{F}\uhr\brho$ is shallow---as we explain below, there is more than one canonical common $\ell$-partial decision tree for a {sequence} $\mathscr{F}$ of CNFs.
 
Let us  explain what a canonical common $\ell$-partial decision tree for a {sequence} of CNFs $\mathscr{F}$ is. We will see that there is a set of canonical common $\ell$-partial decision trees for a given $\mathscr{F}$ rather than just one tree; note that this is the case even though we assume a fixed ordering $F_1,F_2,\dots$ on the elements of $\mathscr{F}$ as well as on the clauses within each CNF.  (Observe the contrast with the case of a canonical decision tree for a single formula $F$, where we assume a fixed ordering on the clauses of $F$; in that setting, as explained above there is a single canonical decision tree $\CDT(F).$)

We need a preliminary definition to handle a technical issue related to the final segment of paths through a canonical decision tree.

\begin{definition}[Full paths in the CDT]
\label{def:full}
Let $F = C_1 \wedge C_2 \wedge \cdots$ be a $k$-CNF and consider the canonical decision tree $\CDT(F)$ for $F$. Every path $\eta$ in $\CDT(F)$ can be written as the the disjoint union of segments $\eta = \eta^{(1)} \circ \eta^{(2)} \circ \cdots \circ \eta^{(u)}$,  where for all $j\in [u]$, the segment $\eta^{(j)}$ is an assignment to the surviving variables in the restricted clause $C_{i_j} \uhr \eta^{(1)} \circ \cdots \circ \eta^{(j-1)}$, and $C_{i_j}$ is the first clause in $F \uhr \eta^{(1)} \circ \cdots \circ \eta^{(j-1)}$ that is not identically-$1$.

Furthermore, note that for $j\in [u-1]$, the segment $\eta^{(j)}$ is in fact an assignment fixing \emph{all} the surviving variables in $C_{i_j} \uhr \eta^{(1)} \circ \cdots \circ \eta^{(j-1)}$.  We say that $\eta$ is \emph{full} if this is also the case for the final segment: $\eta$ is \emph{full} if $\eta^{(u)}$ is an assignment fixing all the surviving variables in $C_{i_u} \uhr \eta^{(1)} \circ \cdots \circ \eta^{(u-1)}$. 
\end{definition} 

\begin{observation}
Let $F$ be a $k$-CNF and suppose $\depth(\CDT(F)) > \ell$. Then there is a full path $\eta$ of length $|\eta| \in \{\ell+1,\ldots,\ell+k\}$ in $\CDT(F)$. 
\end{observation} 

To help minimize confusion, we will reserve ``$\eta$'' for paths or segments of paths in CDTs, and ``$\pi$'' for paths (or segments of paths) in CCDTs.

We are now ready to define the set of canonical common $\ell$-partial decision trees:

\begin{definition}[Canonical common $\ell$-partial DT]
\label{def:ccpdt}
Let $\mathscr{F} = (F_1,\ldots,F_M)$ be an ordered collection of $k$-CNFs.  The set of all \emph{canonical} common $\ell$-partial decision trees for $\mathscr{F}$, which we denote $\CCDT_\ell(\mathscr{F})$, is defined inductively as follows: 
\begin{enumerate}

\item [0.] If $M=0$ (i.e. $\mathscr{F}$ is an empty collection of $k$-CNFs) then $\CCDT_\ell(\mathscr{F})$ contains a single tree, the empty tree with no nodes.  (Note that otherwise $M\geq 1$, so there is some first formula $F_1$ in $\mathscr{F}.$)

\item If $\CDT(F_1) \le \ell$, then $\CCDT_\ell(\mathscr{F})$ is simply $\CCDT_\ell(\mathscr{F}')$, where $\mathscr{F}' = (F_2,\dots,F_M)$. (Note that in this case, since inductively each tree in $\CCDT_\ell(\mathscr{F}')$ is a common $\ell$-partial DT for $\mathscr{F}'$, each such tree is also a common $\ell$-partial DT for $\mathscr{F}$.)

\item Otherwise, since $\CDT(F_1) > \ell$ there must be a witnessing {\emph{full}} path {$\eta$} of length {between $\ell+1$ and $\ell + k$} in $\CDT(F_1)$, and there are at most {$2^{\ell + k}$} such witnessing full paths. Let $P$ be the set of all such witnessing full paths.  For each path ${\eta} \in P$,  let $T_{\eta}$ be the tree of depth {$|\eta|$} obtained by exhaustively querying all the variables in {$\eta$} in the first {$|\eta|$} levels.  Recurse at the end of each path in $T_{{\eta}}$: for each path $\pi$ in $T_{{\eta}}$, attach a tree $T'$ from $\CCDT_\ell(\mathscr{F}\uhr \pi)$ at the end of the path.  So in this case $\CCDT_\ell(\mathscr{F})$ is the set of all trees that can be obtained in this way (across all possible choices of ${\eta} \in P$ and all possible choices of a tree $T' \in \CCDT_\ell(\mathscr{F}\uhr \pi)$ for each path $\pi \in T_{{\eta}}$).

\end{enumerate} 
We write $\depth(\CCDT_\ell(\mathscr{F}))$ to denote the maximum depth of any tree in the set 
$\CCDT_\ell(\mathscr{F})$.
\end{definition}

The following slight variant of Theorem~\ref{thm:HSL14} can be extracted, with some effort, from a slight modification of the proof given in~\cite{Has14}, which we provide in Appendix~\ref{sec:HSL14-proof}:

\begin{theorem}[Slight variant of H{\aa}stad's multi-switching lemma. Theorem~\ref{thm:HSL14}]
\label{thm:HSL14-canonical}
Let $\mathscr{F} = (F_1,\ldots,F_M)$ be an ordered collection of $k$-CNFs. Then for all $\ell, t \ge 1$, 
\[ \Prx_{\brho\leftarrow\calR_p}[\, \depth(\CCDT_{\ell}(\mathscr{F}\uhr\brho)) \ge t\,] \le M^{\lceil t/\ell \rceil}(32pk)^t.  \] 
\end{theorem}

\vspace{-8pt} 
\paragraph{A comparison of Theorem~\ref{thm:HSL14} (H{\aa}stad's multi-switching lemma) and Theorem~\ref{thm:HSL14-canonical} (our variant of it).}  We emphasize that the differences are technical in nature, and all the ideas in our proof of Theorem~\ref{thm:HSL14-canonical} are from~\cite{Has14}. First, we observe that $\ell$ is now a free parameter rather than being fixed to $\log(2M)$; this flexibility will be necessary in our PRG construction for sparse $\F_2$ polynomials (where we take $\ell = \Theta(\sqrt{\log M})$).  Second, our notion of a canonical common partial decision tree differs slightly from the one that is implicit in~\cite{Has14}: in case 2 of Definition~\ref{def:ccpdt}, we query a witnessing full path of length between $\ell+1$ and $\ell+k$, whereas~\cite{Has14} queries any witnessing path of length greater than $\ell$.

\section{A pseudorandom multi-switching lemma} \label{sec:derand-H14}

As suggested earlier, the crux of our PRG construction is a \emph{derandomization} of the multi-switching lemma of Theorem~\ref{thm:HSL14-canonical}: we devise a suitable \emph{pseudorandom} distribution over random restrictions in place of $\calR_p$ (the truly random distribution over restrictions) and show that a random restriction $\brho$ drawn from this pseudorandom distribution satisfies a similar guarantee to Theorem~\ref{thm:HSL14-canonical}.

Our derandomization of Theorem~\ref{thm:HSL14-canonical} is largely influenced by Trevisan and Xue's~\cite{TX13} ingenious derandomization of H{\aa}stad's original switching lemma (Theorem~\ref{thm:HSL}).  Roughly speaking, we will derandomize the multi-switching lemma of Theorem~\ref{thm:HSL14-canonical} by ``fooling its proof'':  we will show that the proof of Theorem~\ref{thm:HSL14-canonical} (given in Appendix~\ref{sec:HSL14-proof}, which we again emphasize is only a slight technical modification of H{\aa}stad's proof of his multi-switching lemma, Theorem~\ref{thm:HSL14}) ``cannot $\delta$-distinguish" between truly random restrictions and pseudorandom restrictions drawn from $\polylog(n)$-wise independent distributions.  Since Theorem~\ref{thm:HSL14-canonical} holds for truly random restrictions, it thus follows that it also holds for pseudorandom restrictions drawn from $\polylog(n)$-wise independent distributions (up to a $\delta$ additive loss in the failure probability). 

To accomplish this, we exploit the ``computational simplicity" of Theorem~\ref{thm:HSL14-canonical}'s proof: for a fixed family $\mathscr{F}$ of $k$-CNF formulas, we will show that there is a small $\acz$ circuit that takes as input an encoding of a restriction $\rho$, and outputs $1$ iff $\rho$ is a bad restriction for the desired conclusion of Theorem~\ref{thm:HSL14-canonical}, contributing to its failure probability (i.e.~iff $\depth(\CCDT_\ell(\mathscr{F} \uhr \rho)) > t$).   As alluded to in Section~\ref{sec:ccdt}, this relies on the fact that Theorem~\ref{thm:HSL14-canonical} does not simply bound the depth of the \emph{optimal} common $\ell$-partial decision tree for $\mathscr{F} \uhr \brho$, but instead the depth of any \emph{canonical} common $\ell$-partial decision tree for $\mathscr{F} \uhr \brho$. Indeed, this ``constructive" aspect of the proof  is crucial for our derandomization strategy: it is not at all clear that there is a small circuit for checking if the \emph{optimal} common $\ell$-partial decision tree for $\mathscr{F} \uhr \rho$ has depth greater than $t$.

\medskip 

It will be convenient for us to represent restrictions $\rho \in \{0,1,\ast\}^n$ as bitstrings $(\varrho,y) \in \zo^{n\times q} \times \zo^n := \{0,1\}^{Y_q}$, where $q \in \N$ is a parameter. \ignore{(the semantics of $q$ is explained in Observation~\ref{obs:string-to-restriction} below).}

\begin{definition}[Representing restrictions as bitstrings] 
\label{def:string-to-restriction}
We associate with each string $(\varrho,y) \in \zo^{Y_q}$ the restriction $\rho(\varrho,y) \in \{0,1,\ast\}^{n}$ defined as follows: 
\[ 
\rho(\varrho,y)_i  = \begin{cases}
\ast & \text{if $\varrho_{i,1} = \cdots = \varrho_{i,q} = 1$} \\
y_i & \text{otherwise}. 
\end{cases}
\] 
\end{definition} 

The following observation explains the role of $q$:

\begin{observation}
\label{obs:string-to-restriction}
Let $(\bvrho,\by)$ be drawn from the uniform distribution over $\zo^{Y_q}$.  Then the random restriction $\rho(\bvrho,\by) \in \{0,1,\ast\}^n$ is distributed according to $\calR_p$ where $p = 2^{-q}$.
\end{observation}

Our main result in this section is a pseudorandom multi-switching lemma:

\begin{theorem}[Derandomized version of Theorem~\ref{thm:HSL14-canonical}]
\label{thm:derand-H14}
Let $\mathscr{F} = (F_1,\ldots,F_M)$ be an ordered list of $Q$-clause $k$-CNFs. Let $\delta,p \in (0,1)$ and define $q = \log(1/p)$. Let $\calD$ be any distribution over $\zo^{Y_q}$ that $(\delta/(M^{\lceil t/\ell\rceil} n^{O(t)}))$-fools the class of depth-$3$ circuits of size $M(n^{O(\ell)} + Q  2^{O(kq)})$.  Then for all $\ell \geq k$ and all $t\in \N$,
 \[ \Prx_{(\boldeta,\bz)\leftarrow\calD}\big[\, \depth(\CCDT_{\ell}(\mathscr{F}\uhr\rho(\boldeta,\bz))) \ge t\,\big] \le 16^{t+\ell}M^{\lceil t/\ell\rceil} (32pk)^t + \delta. \]
 \end{theorem}

\subsection{Bad restrictions and the structure of witnessing paths} 
\label{sec:bad-restrictions}

Fix $\mathscr{F} = (F_1,\ldots,F_M)$.  We say that a restriction $\rho \in \{0,1,\ast\}^n$ is \emph{bad}  if 
\[ \depth(\CCDT_{\ell}(\mathscr{F}\uhr\rho)) \ge t. \] 
Fix $\rho$ to be a bad restriction.
Recalling our definition of the set of canonical common partial decision trees (Definition~\ref{def:ccpdt}), there exists a tree $T \in \CCDT_\ell(\mathscr{F} \uhr \rho)$ and a path $\Pi$ of length exactly $t$ through $T$. Furthermore, we have that 
\begin{enumerate} 
\item There exist indices $1\le i_1 \le i_2 \le \cdots \le i_u \le M$ where $u \le \lceil t/\ell\rceil$, and 
\item $\Pi = \pi^{(1)} \circ \cdots \circ \pi^{(u)}$, where for all $j\in [u]$, we have that $\supp(\pi^{(j)}) =  \supp(\eta^{(j)})$ where  $\eta^{(j)}$ is a path through the canonical decision tree 
\[ \CDT(F_{i_j} \uhr \rho \circ \pi^{(1)}\circ\cdots\circ \pi^{(j-1)}).\]  
Furthermore, for every $j\in [u-1]$ we have that $\eta^{(j)}$ is a full path of length between $\ell+1$ and $\ell+k$ through the CDT, and $\eta^{(u)}$ is a path of length exactly $t - \sum_{j=1}^{u-1} |\supp(\eta^{(j)})|$. (Note that $\eta^{(u)}$ is not necessarily a full path.)

\end{enumerate}  
(Note that by (2), these subpaths $\pi^{(j)}$ of $\Pi$ are supported on mutually disjoint sets of coordinates.) With this structure of  $\Pi$ in mind, we make the following definition: 

\begin{definition}[$\mathscr{F}$-traversal]
\label{def:F-traversal}
Let $\mathscr{F} = (F_1,\ldots,F_M)$ be an ordered list of CNFs.  An \emph{$\ell$-segmented $\mathscr{F}$-traversal of length $t$} is a tuple $P = (\mathscr{I}, \{S_1,\ldots,S_u\},\Pi,{\mathrm{H}})$  comprising: 
\begin{enumerate}
\item An ordered list of indices $\mathscr{I} = (i_1, \ldots, i_u)$ where $1\le i_1 \le \cdots \le i_u \le M$ and $u \le \lceil t /\ell \rceil$,
\item For each index $i_j \in \mathscr{I}$,  a subset $S_j \sse [n]$ such that 
\begin{enumerate}
\item These sets are mutually disjoint: $S_j \cap S_{j'} = \emptyset$ for all $j\ne j'$.
\item For $1\le j \le u-1$, each $S_j$ has size between $\ell+1$ and $\ell+k$, and $S_u$ has size exactly $t - \sum_{j=1}^{u-1} |\supp(\eta^{(j)})|$.
\end{enumerate} 
(Consequently $|S_1\cup\cdots\cup S_u| =t$.)
\item An assignment $\Pi = \pi^{(1)}\circ\cdots\circ\pi^{(u)}$ to the variables in $S_1 \cup \cdots \cup S_u$, where 
\[ \pi^{(j)} : \zo^{S_j} \to \zo \qquad \text{for $1\le j\le u$.} \] 
\item An assignment $\mathrm{H} = \eta^{(1)}\circ \cdots\circ\eta^{(u)}$ to the variables in $S_1 \cup \cdots \cup S_u$, where again
\[ \eta^{(j)} : \zo^{S_j} \to \zo \qquad \text{for $1\le j\le u$.} \] 
\end{enumerate} 
\end{definition} 

By our discussion above, for any restriction $\rho \in \{0,1,\ast\}^n$ and any tree $T \in \CCDT_\ell(\mathscr{F} \uhr \rho)$, every path $\Pi$ of length $t$ through $\CCDT_{\ell}(\mathscr{F}\uhr \rho)$ uniquely induces an $\ell$-segmented $\mathscr{F}$-traversal $P$ of length $t$.  We say that \emph{$P$ occurs in $\CCDT_{\ell}(\mathscr{F}\uhr\rho)$} if it is induced by some path $\Pi$ of length $t$ through $T$ for some $T \in \CCDT_\ell(\mathscr{F} \uhr \rho)$.

Definition~\ref{def:F-traversal} immediately yields the following: 

\begin{proposition}[Number of $\mathscr{F}$-traversals]
\label{prop:number-of-traversals} 
Fix an ordered list $\mathscr{F} = (F_1,\ldots,F_M)$ of $k$-CNFs, and let $\mathcal{P}_{\mathscr{F},\ell,t}$ denote the collection of all $\ell$-segmented $\mathscr{F}$-traversals of length $t$. Then 
\[ |\mathcal{P}_{\mathscr{F},\ell,t}| \le  M^{\lceil t/\ell \rceil} n^{O(t)}. \] 
\end{proposition} 


\subsection{A small $\acz$ circuit for recognizing bad restrictions} \label{sec:small-acz-bad}

We begin by showing that for every $\mathscr{F}$-traversal $P = (\mathscr{I}, \{S_1,\ldots,S_u\}, \Pi,{\mathrm{H}})$, there is a small circuit $\calC_P$ over $\zo^{Y_q}$ that outputs $1$ on input $(\varrho,y)\in \zo^{Y_q}$ iff $P$ occurs in $\CCDT_\ell(\mathscr{F}\uhr \rho(\varrho,y))$.  Since
\begin{align*}
\text{$\rho(\varrho,y)$ is bad} &\Longleftrightarrow \depth(\CCDT_\ell(\mathscr{F} \uhr \rho(\varrho,y))) \ge t \\
&\Longleftrightarrow \text{$\exists$ $\ell$-segmented $\mathscr{F}$-traversal $P$ of length $t$ occurring in $\CCDT_\ell(\mathscr{F}\uhr\rho(\varrho,y))$},
\end{align*}
by considering  
\begin{equation}
   \calC_{\mathscr{F},\ell,t}(\varrho,y) := \bigvee_{P\in\calP_{\mathscr{F},\ell,t}} \calC_P(\varrho,y) \label{eq:big-or} 
   \end{equation} 
we have that 
\[ \text{$\rho(\varrho,y)$ is bad $\Longleftrightarrow \calC_{\mathscr{F},\ell,t}(\varrho,y) = 1$.}
 \] 
\begin{claim}[Circuit for a single $\mathscr{F}$-traversal]
\label{claim:subcircuit} 
Let $P = (\mathscr{I}, \{S_1,\ldots,S_{u}\}, \Pi, {\mathrm{H}})$ be an $\ell$-segmented $\mathscr{F}$-traversal of length $t$. There is a depth-$3$ AND-OR-AND circuit $\calC_P : \zo^{Y_q} \to \zo$ of size $M(n^{O(\ell)}+Q2^{O(kq)})$ such that
\[ \forall\,(\varrho,y)\in \zo^{Y_q}\colon\ \ \calC_P(\varrho,y) =1 \Longleftrightarrow \text{$P$ occurs in $\CCDT_\ell(\mathscr{F}\uhr\rho(\varrho,y))$} \] 
\end{claim}

\begin{proof}
Our circuit $\calC_P$ will be the AND of $M$ many depth-$3$ subcircuits of size $n^{O(\ell)}$, one for each $k$-CNF $F \in \mathscr{F}$. As we will explain later, each of these subcircuits is one of two types. We first describe these two types of ``candidate subcircuits'', and then explain precisely which $M$ subcircuits of each type are AND-ed together to give $\calC_P$. (Both these types of circuits are implicit in the work of~\cite{TX13}.) 

\begin{enumerate}
\item {\bf First type: Circuits  checking that a particular restriction $\eta$ is a path in a particular CDT.}  We claim that for any $Q$-clause $k$-CNF $F' = C_1 \wedge \cdots \wedge C_Q$ and restriction $\eta$, there is a $Q2^{O(kq)}$-clause $O(kq)$-CNF $G$ over $\zo^{Y_q}$ that outputs $1$ on input $(\varrho,y)$ iff $\eta$ is a path in $\CDT(F' \uhr \rho(\varrho,y))$. 

For each $i\in [Q]$, we write $\Fixed_i$ to denote the set 
\[ \{ j \in [n] \colon \text{$j \in \eta^{-1}(\{0,1\})$ and $x_j$ occurs in $C_i$}\}\]
of all variables that are fixed by $\eta$ and occur in $C_i$.  We write  $\sigma^{(i)} \in \zo^{\Fixed_i}$ to denote $\eta$ restricted to the coordinates in $\Fixed_i$.  It is straightforward to verify that $\eta$ is a path in $\CDT(F'\uhr \rho(\varrho,y))$ iff for all $i\in [Q]$ such that $\Fixed_1 \cup \cdots \cup \Fixed_{i-1} \subsetneq \supp(\eta)$, 
\begin{enumerate}
\item If $\Fixed_i \setminus (\Fixed_1 \cup \cdots \cup \Fixed_{i-1}) = \emptyset$ then the clause $C_i$ is satisfied by $\rho(\varrho,y) \circ \sigma^{(1)} \circ \cdots \circ \sigma^{(i-1)}$. (Hence this clause does not contribute to $\CDT(F' \uhr \rho(\varrho,y))$; it is ``skipped" in the canonical decision tree construction process.) 
\item Otherwise, writing $\Fixed_i' := \Fixed_i \setminus (\Fixed_1 \cup \cdots \cup \Fixed_{i-1})$, 
\begin{enumerate}
\item $\rho(\varrho,y)_j = \ast$ for all $j \in \Fixed_i'$, and 
\item $\rho(\varrho,y) \circ \sigma_1 \circ \cdots \circ \sigma_{i-1}$ falsifies all the remaining literals in $C_i$ and are not in $\Fixed_i'$.
\end{enumerate} 
In other words, the clause  
\[ C_i \uhr \rho(\varrho,y) \circ \sigma^{(1)} \circ \cdots\circ \sigma^{(i-1)} \]
is not satisfied and its surviving variables are precisely those in $\Fixed_i'$. (Hence the variables in $\Fixed_i'$ are exactly those queried by the canonical decision tree construction process when it reaches $C_i$.) 
\end{enumerate} 
Since both conditions (a) and (b) depend only on the coordinates of $\rho(\varrho,y)$ that occur in $C_i$ (at most $k$ such coordinates since $C_i$ has width at most $k$), and hence at most $k(q+1)$ coordinates of $(\varrho,y) \in \zo^{Y_q}$, it is clear that both conditions can be checked by a $2^{O(kq)}$-clause $O(kq)$-CNF over $\zo^{Y_q}$.  The overall CNF $G$ is simply the AND of all $Q$ many of these CNFs, one for each clause $C_i$ of $F'$, and hence $G$ is itself a $Q2^{O(kq)}$-clause $O(kq)$-width CNF.

\item {\bf Second type: Circuits checking that a particular $\CDT$ has depth at most $\ell$.} Next, we claim that for every $Q$-clause $k$-CNF $F'$, there is a depth-3 AND-OR-AND circuit with fan-in sequence $((2n)^{\ell+1}, Q2^{O(kq)}, O(kq))$ that outputs $1$ on input $(\varrho,y)$ iff $\depth(\CDT(F' \uhr \rho(\varrho,y))) \le \ell$. 

We establish this by showing that there is a depth-$3$ OR-AND-OR circuit $\Sigma$ with the claimed fan-in sequence that outputs $1$ on input $(\varrho,y)$ if $\depth(\CDT(F' \uhr\rho(\varrho,y))) > \ell$; given such a circuit $\Sigma,$ the desired AND-OR-AND circuit is obtained by negating $\Sigma$ and using de Morgan's law. Certainly $\depth(\CDT(F'\uhr \rho(\varrho,y))) > \ell$ iff there is a path $\eta$ of length $\ell+1$ in $\CDT(F' \uhr\rho(\varrho,y))$.  There are at most $(2n)^{\ell+1}$ many possible paths of length $\ell+1$ (every path is simply an ordered list of literals), and as argued in (1) above, for every path $\eta$ there is a $Q2^{O(kq)}$-clause, $O(kq)$-CNF over $\zo^{Y_q}$ that checks if $\eta$ is a path in $\CDT(F'\uhr \rho(\varrho,y))$.  The overall circuit $\Sigma$ is simply the OR of at most $(2n)^{\ell+1}$ such circuits, one for each path $\eta$. 
\end{enumerate}

With these two types of circuits in hand the overall circuit $\calC_P$ is now easy to describe. $\calC_P$ is the AND of $M$ many depth-$3$ subcircuits, one for each $k$-CNF $F \in \mathscr{F}$: 
\begin{itemize}
\item For each of the $u$ indices $i_j \in \mathscr{I}$, a circuit of the first type that checks that $\eta^{(j)}$ is a path in $\CDT(F_{i_j} \uhr \rho(\varrho,y) \circ \pi^{(1)} \circ \cdots \circ \pi^{(j-1)})$ (recall from Definition~\ref{def:F-traversal} that $\eta^{(j)}$ is ${\mathrm{H}}$ restricted to the variables in $S_j$); 
\item For all $M- u$ other indices $i \in [M] \setminus \mathscr{I}$, a circuit of the second type that checks that $\depth(\CDT(F_i \uhr \rho(\varrho,y) \circ \pi^{(1)}\circ \cdots\circ\pi^{(i^{-})})) \le \ell$, where $i^{-} = \max \{ j \in [u] \colon i_j < i\}$.
\end{itemize} 
The bound on the size of this overall circuit follows from a union bound over the sizes of the subcircuits given in (1) and (2) above. 
\end{proof} 


\subsection{Putting the pieces together: Proof of Theorem~\ref{thm:derand-H14}}

Recalling the definition (\ref{eq:big-or}) of $\calC_{\mathscr{F},\ell,t}$, 
\[ \calC_{\mathscr{F},\ell,t}(\varrho,y) := \bigvee_{P\in\calP_{\mathscr{F},\ell,t}} \calC_P(\varrho,y), \] 
Proposition~\ref{prop:number-of-traversals} giving a bound on its top fan-in, and Claim~\ref{claim:subcircuit} giving a bound on the size of its subcircuits, we have shown the following: 

\begin{claim}[Circuit for recognizing bad restrictions]
\label{claim:overall-ckt}  
Let $\mathscr{F} = (F_1,\ldots,F_M)$ be an ordered list of $Q$-clause $k$-CNFs, and let $\ell,t\ge 1$.  There is a depth-$4$ circuit $\calC_{\mathscr{F},\ell,t}$ over $\zo^{Y_q}$ such that 
\[ \calC_{\mathscr{F},\ell,t}(\varrho,y) = 1 \quad \Longleftrightarrow \quad \depth(\CCDT_\ell(\mathscr{F} \uhr \rho(\varrho,y))) \ge t. \] 
This circuit $\calC_{\mathscr{F},\ell,t}$ is the OR of $M^u n^{O(t)}$ many depth-$3$ circuits of size $M(n^{O(\ell)} + Q2^{O(kq)})$. 
\end{claim}

The following observation will be useful for us: 

\begin{observation}
\label{obs:few-sat} 
Let $\mathscr{F} = (F_1,\ldots,F_M)$ be an ordered collection of $k$-CNFs.  For $\ell \geq k$, the total number of paths $\Pi$ such that $\Pi$ is a path of length exactly $t$ in some tree $T \in \CCDT_\ell(\mathscr{F})$ is at most $(2^{\ell+k} \cdot 2^{\ell + k})^{\lceil t/\ell\rceil} \le 16^{t+\ell}$. 
Consequently, if $(\varrho,y)\in \zo^{Y_q}$ is such that $\calC_{\mathscr{F},\ell,t}(\varrho,y) = 1$, then $\calC_P(\varrho,y) =1$ for (at least one) and at most $16^{t+\ell}$ many $\ell$-segmented $\mathscr{F}$-traversals $P$ of length $t$. 
\end{observation}

\begin{proof}
This follows by inspection of the recursive construction of the set $\CCDT_\ell(\mathscr{F})$ of canonical common $\ell$-partial decision trees for $\mathscr{F}$.  Each time case (2) of the definition is reached, the set $P$ of witnessing full paths has size at most $2^{\ell + k}$, and for each path in $P$ there are at most $2^{\ell + k}$ possible assignments to the variables on the path.  Finally, there are at most $\lceil t/\ell \rceil$ levels of recursive calls. 
\end{proof}

With Claim~\ref{claim:overall-ckt} and Observation~\ref{obs:few-sat} in hand, we are now ready to prove our main result of this section (Theorem~\ref{thm:derand-H14}), a derandomized version of the multi-switching lemma (Theorem~\ref{thm:HSL14-canonical}). We restate Theorem~\ref{thm:derand-H14} here for the reader's convenience:

\begin{reptheorem}{thm:derand-H14} 
Let $\mathscr{F} = (F_1,\ldots,F_M)$ be an ordered list of $Q$-clause $k$-CNFs. Let $\delta,p \in (0,1)$ and define $q = \log(1/p)$. Let $\calD$ be any distribution over $\zo^{Y_q}$ that $(\delta/(M^{\lceil t/\ell\rceil} n^{O(t)}))$-fools the class of depth-$3$ circuits of size $M(n^{O(\ell)} + Q  2^{O(kq)})$.  Then for all $\ell \geq k$ and all $t\in \N$,
 \[ \Prx_{(\boldeta,\bz)\leftarrow\calD}\big[ \depth(\CCDT_{\ell}(\mathscr{F}\uhr\rho(\boldeta,\bz))) \ge t\big] \le 16^{t+\ell}M^{\lceil t/\ell\rceil} (32pk)^t + \delta. \]
\end{reptheorem}

\begin{proof} 
\begin{align*}
&\Prx_{(\boldeta,\bz)\leftarrow\calD}\big[\,\depth(\CCDT_\ell(\mathscr{F}\uhr\rho(\boldeta,\bz))) \ge t\,\big] \\ 
&= \Ex_{(\boldeta,\bz)\leftarrow\calD}\big[\,\calC_{\mathscr{F},\ell,t}(\boldeta,\bz)\,\big] \tag*{(Claim~\ref{claim:overall-ckt})}\\
&\le \sum_{P \in \calP_{\mathscr{F},\ell,t}}\Ex_{(\boldeta,\bz)\leftarrow\calD}\big[\,\calC_P(\boldeta,\bz)\,\big]  \tag*{(union bound)} \\
&\le \sum_{P \in \calP_{\mathscr{F},\ell,t}}\left(\Ex_{(\bvrho,\by)\leftarrow\calU}[\,\calC_P(\bvrho,\by) \,] + \frac{\delta}{M^{\lceil t/\ell\rceil} n^{O(t)}}\right) \tag*{($\calD$ $(\delta/(M^{\lceil t/\ell\rceil} n^{O(t)}))$-fools $\calC_P$)} \\
&\le \delta + \Ex_{(\bvrho,\by)\leftarrow\calU}\left[\sum_{P \in \calP_{\mathscr{F},\ell,t}}\,\calC_P(\bvrho,\by) \,\right] \tag*{(Proposition~\ref{prop:number-of-traversals} 
)} \\
&\le \delta + 16^{t+\ell} \Ex_{(\bvrho,\by)\leftarrow\calU}[\,\calC_{\mathscr{F},\ell,t}(\bvrho,\by)\,] \tag*{(Observation~\ref{obs:few-sat})} \\
&= \delta + 16^{t+\ell} \Prx_{(\bvrho,\by)\leftarrow\calU}\big[\,\depth(\CCDT_\ell(\mathscr{F}\uhr\rho(\bvrho,\by))) \ge t\,\big] \tag*{(Claim~\ref{claim:overall-ckt})}\\
&= \delta + 16^{t+\ell} \Prx_{\brho\leftarrow\calR_p}\big[\,\depth(\CCDT_\ell(\mathscr{F}\uhr \brho)) \ge t\,\big] \tag*{(Observation~\ref{obs:string-to-restriction})}  \\ 
&\le \delta + 16^{t+\ell} M^{\lceil t/\ell\rceil} (32pk)^t. \tag*{(Theorem~\ref{thm:HSL14-canonical})} 
\end{align*} 
\end{proof}

\section{Applying our pseudorandom multi-switching lemma: the Ajtai--Wigderson framework for PRG constructions} 
\label{sec:AW} 

Implicit in the early work of Ajtai--Wigderson~\cite{AjtaiWigderson:85} giving the first PRG for $\acz$ circuits is a powerful, generic framework for constructing PRGs from ``pseudorandom simplification lemmas". In this section we give an explicit description of their framework in general terms. Our work shows that this framework is fairly versatile: both our PRGs, for $\acz$ circuits and sparse $\F_2$ polynomials, are obtained within it (albeit with specialized pseudorandom simplification lemmas for each class). Variants of these ideas from~\cite{AjtaiWigderson:85} are also present in the more recent PRG constructions of~\cite{GMRTV12,IMZ12,RSV13,TX13}.

\begin{itemize} 
\item Let $\mathscr{C}$ be the function class of interest, the class for which we would like to design a PRG. For us $\mathscr{C}$ will either be the class of size-$M$ depth-$d$ $\acz$ circuits, or the class of $S$-sparse $\F_2$ polynomials. (Our analysis will assume that $\mathscr{C}$ is closed under restrictions, which holds for natural function classes including our two classes of interest.)
\item Let $\mathscr{C}_\simple$ be a class of ``simple" functions. We will describe the relationship between $\mathscr{C}$ and $\mathscr{C}_\simple$ in detail shortly, but we mention here that this approach relies on the simplicity of the functions in $\mathscr{C}_\simple$ enabling PRGs of short seed length. For us, when $\mathscr{C}$ is the class of $\acz$ circuits, $\mathscr{C}_\simple$ will be the class of small-depth decision trees; when $\mathscr{C}$ is the class of sparse $\F_2$ polynomials, $\mathscr{C}_\simple$ will be the class of small-depth decision trees with low-degree $\F_2$ polynomials at its leaves. (Note that we do not require that $\mathscr{C}_\simple$ be a subclass of $\mathscr{C}$.) 
\end{itemize} 

At a high level, the plan is to give a \emph{randomness-efficient} reduction from the task of fooling $\mathscr{C}$ to that of fooling $\mathscr{C}_\simple$;  we obtain a pseudorandom distribution $\calD$ over $\zo^n$ that fools $\mathscr{C}$ by ``pseudorandomly stitching together" independent copies of a pseudorandom distribution $\calD_\simple$ over $\zo^{n'}$ that fools $\mathscr{C}_\simple$ (for some $n' \ll n$).  In more detail, the plan is to fool $\mathscr{C}$ recursively in stages, where in each stage we employ two pseudorandom constructs: 

\begin{enumerate}
\item A PRG for $\mathscr{C}_\simple$, and 
\item A ``pseudorandom $\mathscr{C}$-to-$\mathscr{C}_\simple$ simplification lemma".

Roughly speaking, such a simplification lemma says the following: there is a pseudorandom distribution $\calR$ over restrictions such that for all $\calC \in \mathscr{C}$, with high probability over $\brho\leftarrow \calR$ the randomly restricted function $\calC \uhr \brho$ belongs to $\mathscr{C}_\simple$.  This pseudorandom distribution $\calR$ over the space of restrictions $\{0,1,\ast\}^n$ should have the following structure:  
\begin{enumerate}
\item The set of ``live" positions $\bL \sse [n]$ (i.e.~the set of $\ast$'s) can be sampled with seed length $s_\SL$.  We write $\bL \leftarrow \calR_\stars$ to denote a draw from this pseudorandom distribution over subsets of $[n]$. 
\item Non-live positions $[n]\setminus \bL$ are filled in independently and uniformly with $\{0,1\}$, and do not count against the seed length $s_\SL$. We write $\brho\leftarrow \zo^{[n]\setminus \bL}$ to denote a draw of such a restriction.
\end{enumerate} 

We will require each subset $L \in \supp(\calR_\stars)$ to have size at least $pn$ for some not-too-small $p \in (0,1)$ (equivalently, we will require $\calR$ to be supported on restrictions that leave at least a $p$ fraction of coordinates unfixed); as we will soon see, this ensures that we ``make good process" in each stage.

The guarantee that we will require of this pseudorandom $\mathscr{C}$-to-$\mathscr{C}_\simple$ simplification lemma is  as follows: for every $\calC \in \mathscr{C}$, 
\begin{equation} \Ex_{\bL\leftarrow\calR_{\stars}}\bigg[ \Prx_{\brho\leftarrow\zo^{[n]\setminus \bL}}\big[\, (\calC \uhr \brho) \notin \mathscr{C}_\simple\big] \bigg]  \le \delta_\SL, \label{eq:PSL}
\end{equation} 
where the failure probability $\delta_\SL$ is as small as possible.  
\end{enumerate} 

\noindent {\bf An aside about applying Theorem~\ref{thm:derand-H14} within this framework.} The astute reader may have noticed that our pseudorandom multi-switching lemma (Thereom~\ref{thm:derand-H14}) from the previous section is established for a distribution over restrictions that does \emph{not} have the structure prescribed above: rather than a pseudorandom choice of live variables $\bL \sse [n]$ and a fully random choice of bits as values for the non-live variables ${[n] \setminus \bL}$, Theorem~\ref{thm:derand-H14} is established for a distribution over restrictions where both choices are pseudorandom. (Recalling Definition~\ref{def:string-to-restriction}, we see that $\boldeta$ in the statement of Theorem~\ref{thm:derand-H14} corresponds to the choice of $\bL \sse [n]$, and $\bz$ to the choice of bits for the coordinates in $[n] \setminus \bL$; in the proof of Theorem~\ref{thm:derand-H14} this pair $(\boldeta,\bz)$ is sampled from a single pseudorandom distribution over $Y_q$.)  However, this suggests that Theorem~\ref{thm:derand-H14} is ``stronger than it has to be", since it is more randomness efficient than necessary for this application. Indeed, in Proposition~\ref{prop:swap} we formalize this intuition, showing that our proof of Theorem~\ref{thm:derand-H14} also extends to hold for distributions over restrictions with the prescribed structure.

\vspace{-5pt}

\paragraph{One stage of the PRG construction.}  Going back to the general framework, we next describe how the two pseudorandom constructs described above---a PRG for $\mathscr{C}_\simple$ and a pseudorandom $\mathscr{C}$-to-$\mathscr{C}_\simple$ simplification lemma---are employed together within a single stage of the PRG construction for $\mathscr{C}$. 

For $L \sse [n]$ let us write $\delta(L)$ to denote the probability $\Prx_{\brho\leftarrow\zo^{[n]\setminus L}}[\, (\calC \uhr \brho) \notin \mathscr{C}_\simple\,]$; by (\ref{eq:PSL}) we have that $\Ex_{\bL\leftarrow\calR_\stars}[\delta(\bL)] \le \delta_\SL$. Fix an $L \sse [n]$. Let $\calD_\simple$ be a distribution that $\delta_\PRG$-fools $\mathscr{C}_\simple$, and suppose $\calD_\simple$ can be sampled with $s_\PRG$ many random bits. A simple but crucial fact from~\cite{AjtaiWigderson:85} is the following: the distribution over $\zo^n$ where 
\begin{enumerate}
\item The coordinates in $[n] \setminus L$ are filled in with uniform random bits;
\item The coordinates in $L$ are filled in according to the pseudorandom distribution $\calD_\simple$,
\end{enumerate}
$(\delta(L) + \delta_\PRG)$-fools $\mathscr{C}$. That is, for all $\calC \in \mathscr{C}$, 
\[ \mathop{\Ex_{\bx\leftarrow\calU}}_{\by\leftarrow\calD_\simple}\big[ \calC(\bx_{[n]\setminus L}, \by_L)\big] = \Ex_{\bx\leftarrow\calU}\big[\calC(\bx)\big] \pm (\delta(L) + \delta_\PRG). \] 
Taking expectations over $\bL \leftarrow\calR_\stars$ and using (\ref{eq:PSL}), we get that 
\begin{equation} \Ex_{\bL\leftarrow\calR_\stars}\left[\,\mathop{\Ex_{\bx\leftarrow\calU}}_{\by\leftarrow\calD_\simple}\big[ \calC(\bx_{[n]\setminus \bL}, \by_{\bL})\big]\right] = \Ex_{\bx\leftarrow\calU}\big[\calC(\bx)\big] \pm (\delta_\SL + \delta_\PRG). \label{eq:bias-preservation}
\end{equation} 
Consider the distribution $\calR_\gentle$ over the space of restrictions $\{0,1,\ast\}^n$ defined as follows: to make a draw $\bpi \leftarrow\calR_\gentle$, first make draws $\bL \leftarrow\calR_\stars$ and $\by \leftarrow\calD_\simple$, and then output the restriction $\bpi \in \{0,1,\ast\}^n$ where 
\[ 
\bpi_i = \begin{cases}
\by_i & \text{if $i \in \bL$} \\
\ast & \text{otherwise.}  
\end{cases} \quad \text{for all $i\in [n]$.}
\] 
In words, $\bpi$ is the restriction that fixes the coordinates in $\bL$ according to $\by$. With this definition of $\calR_\gentle$ in hand, we can rewrite (\ref{eq:bias-preservation}) as 
\begin{equation}
\Ex_{\bpi\leftarrow\calR_\gentle} \bigg[ \Ex_{\bpi\leftarrow\calU}\big[(\calC\uhr\bpi)(\bx)\big]\bigg] =  \Ex_{\bx\leftarrow\calU}\big[\calC(\bx)\big] \pm (\delta_\SL + \delta_\PRG). 
\label{eq:bias-preservation2} 
\end{equation}
Note that a draw $\bpi \leftarrow \calR_\gentle$ can be sampled with $s_\SL + s_\PRG$ random bits. (We need $s_\SL$ random bits to make a draw $\bL \leftarrow \calR_\stars$, and $s_\PRG$ random bits to make a draw $\by\leftarrow\calD_\simple$.)

We emphasize that the restriction $\bpi$ is supported on $\bL$ (i.e.~$\bpi^{-1}(\{0,1\}) = \bL$), rather than $[n] \setminus \bL$.  For this reason we may view $\calR_\gentle$ as being ``dual" to the distribution $\calR$ that yields a $\mathscr{C}$-to-$\mathscr{C}_\simple$ simplification lemma: while $\calR$ is supported on restrictions that leave at least a $p$ fraction of coordinates unfixed, $\calR_\gentle$ is supported on restrictions that fix at least a $p$ fraction of coordinates. This is explains why,  as alluded to above, we  require the pseudorandom simplification lemma to be such that every $L \in \supp(\calR_\stars)$ has size at least $pn$ for some not-too-small $p\in (0,1)$. 

\vspace{-5pt} 

\paragraph{Fooling $\mathscr{C}$ recursively: the overall PRG construction and its analysis.}  We have sketched the construction of a distribution $\calR_\gentle$ over restrictions in $\{0,1,\ast\}^n$ that preserves $\calC$'s bias up to an error of $(\delta_\SL + \delta_\PRG)$ in the sense of (\ref{eq:bias-preservation2}); furthermore, $\calR_\gentle$ is supported on restrictions that fix at least a $p$ fraction of coordinates. Since an $\eps$-PRG is simply a distribution over assignments in $\{0,1\}^n$ that preserves $\calC$'s bias up to an error of $\eps$, we see that we have made a ``$p$-fraction of progress" towards a PRG, while incurring $(\delta_\SL + \delta_\PRG)$ out of the total $\eps$ amount of error allowed. 

Our PRG construction will recurse on $\calC \uhr \pi$ for all $\pi \in \supp(\calR_\gentle)$, all of which are functions over at most $(1-p)n$ variables. (Since $\mathscr{C}$ is closed under restrictions, we note that $\calC \uhr \pi$ belongs to $\mathscr{C}$ and so we can indeed apply the same argument recursively.)  By fixing at least a $p$ fraction of the remaining coordinates in each stage, we ensure that there are at most $p^{-1}\ln n$ stages in total, after which $n$ coordinates will have been fixed. Hence, as long as 
\[ \delta_\SL + \delta_\PRG \le \frac{\eps}{p^{-1}\ln n}, \] 
i.e.~the total error incurred across all stages is at most $\eps$, we will have that the final distribution over $\zo^n$ does indeed $\eps$-fool $\calC$. 

As noted above, the seed length required to sample from $\calR_\gentle$ in each stage is $s_\SL + s_\PRG$. Since there are at most $p^{-1}\ln n$ stages in total, the overall seed length of this PRG construction is 
\[ (s_\SL + s_\PRG) \cdot p^{-1}\ln n.\]

The following theorem summarizes the upshot of our discussion in this section: 

\begin{theorem}[PRGs from pseudorandom simplification lemmas; implicit in~\cite{AjtaiWigderson:85}]
\label{thm:AW}
Let $\mathscr{C}$ and $\mathscr{C}_\simple$ be two function classes over $\zo^n$, and suppose we have 
\begin{enumerate}
\item A $\delta_\PRG$-PRG for $\mathscr{C}_\simple$ with seed length $s_\PRG(\delta_\PRG)$ for all $\delta_\PRG > 0$, and 
\item A pseudorandom $\mathscr{C}$-to-$\mathscr{C}_\simple$ simplification lemma with the following parameters: for all $\delta_\SL > 0$, there is a distribution $\calR_\stars$ over subsets of $[n]$ such that 
\begin{enumerate}
\item A draw $\bL \leftarrow\calR_\stars$ can be sampled with $s_\SL(\delta_\SL)$ random bits. 
\item Every $L \in \supp(\calR_\stars)$ satisfies $|L| \ge pn$ for some $p \in (0,1)$. 
\item For all $\calC \in \mathscr{C}$, we have that 
 \[ \Ex_{\bL\leftarrow\calR_{\stars}}\bigg[ \Prx_{\brho\leftarrow\zo^{[n]\setminus \bL}}\big[\, (\calC \uhr \brho) \notin \mathscr{C}_\simple\big] \bigg]  \le \delta_\SL.\] 
 \end{enumerate} 
\end{enumerate} 
Then for all $\eps > 0$, there is an $\eps$-PRG for $\mathscr{C}$ with seed length
\[ \left(s_\SL\left(\frac{\eps p}{2\ln n}\right) + s_\PRG\left(\frac{\eps p}{2\ln n}\right)\right) \cdot p^{-1}\ln n.\] 
\end{theorem}

\section{Pseudorandom simplification lemmas for $\acz$ circuits and sparse $\F_2$ polynomials}  \label{sec:prsl}

In order to apply Theorem~\ref{thm:derand-H14}, we need a PRG that can fool depth-3 circuits (to play the role of $\calD$ in that theorem).
We recall a very recent result of Harsha and Srinivasan giving the first PRG for fooling $\acz$ with a seed length whose $\eps$-dependence is $\log(1/\eps)$; we state this result, specialized to the notation of Section~\ref{sec:derand-H14}, below.
\begin{theorem}[\cite{HS16}]
\label{thm:HS16}
The class of size-$S$ depth-$d$ circuits over $\{0,1\}^{Y_q}$ is $\delta$-fooled by $r_{\mathrm{HS}}$-wise independence where 
\[ r_{\mathrm{HS}}(S,d,\delta) = \log^{3d+O(1)}(S) \cdot \log(1/\delta). \]
\end{theorem}

We will need an elementary fact that states, roughly speaking, that if $\calD$ is a distribution that fools a class $\mathscr{F}$, then the distribution obtained by replacing a subset of its coordinates with fully random bits also fools $\mathscr{F}$. Specialized to our context, we state this fact as follows:

\begin{proposition}
\label{prop:swap} 
Let $\calD_{r\text{-}\mathrm{wise}}$ be an $r_{\mathrm{HS}}$-wise independent distribution over $\{0,1\}^{Y_q}$ where $r_{\mathrm{HS}}(S,d,\delta)$ is as defined in Theorem~\ref{thm:HS16}.  Consider the distribution $\calD_{\mathrm{mix}}$ over $\{0,1\}^{Y_q}$ where a draw from $\calD_{\mathrm{mix}}$ is $(\boldeta,\by) \in \zo^{n\times q}\times \zo^n$ where 
\begin{enumerate}
\item (Pseudorandom stars) $\boldeta$ is drawn from the marginal distribution of $\calD_{r\text{-}\mathrm{wise}}$ on $\zo^{n\times q}$, and 
\item (Non-stars filled in fully randomly) $\by$ is an independent uniform string drawn from $\zo^n$. 
\end{enumerate} 
Then like $\calD_{r\text{-}\mathrm{wise}}$, this distribution $\calD_{\mathrm{mix}}$ also $\delta$-fools the class of size-$S$ depth-$d$ circuits over $\{0,1\}^{Y_q}$.
\end{proposition}

\begin{proof}
This follows from the same simple argument that gives Fact 9 of~\cite{TX13}. 
\end{proof} 

We can now state the pseudorandom multi-switching lemma that we will use for both our pseudorandom simplification lemmas (for $\acz$ circuits and for sparse $\F_2$ polynomials): 

\begin{lemma}[Stars chosen pseudorandomly, non-stars filled in fully randomly] 
\label{lem:PSL} 
Let $\mathscr{F} = (F_1,\ldots,F_M)$ be an ordered list of $Q$-clause $k$-CNFs. Let $\ell \geq k$, $t \in \N$ and $\delta,p \in (0,1)$, and define $q = \log(1/p)$.  There is a distribution $\calR_\stars$ over subsets of $[n]$ such that the following hold: 
\begin{enumerate} 
\item A draw $\bL\leftarrow\calR_\stars$ can be sampled with $O(r\log n)$ random bits, where 
\[ r = r_{\mathrm{HS}}\left(M\big(n^{O(\ell)} + Q  2^{O(kq)}\big), 3, \frac{\delta}{M^{\lceil t/\ell\rceil} n^{O(t)}}\right)\]
and $r_{\mathrm{HS}}(\cdot,\cdot,\cdot)$ is as defined in Theorem~\ref{thm:HS16}.  
\item $\calR_\stars$ is $p$-regular:  $\Prx_{\bL\leftarrow\calR_\stars}\big[i\in \bL\big] = p$ for all $i\in [n]$. 
\item A multi-switching lemma holds with respect to $\calR_\stars$: 
 \begin{equation} \mathop{\Prx_{\bL\leftarrow\calR_\stars}}_{\brho\leftarrow\zo^{[n]\setminus \bL}}\big[\,\depth(\CCDT_{\ell}(\mathscr{F}\uhr\brho)) \ge t\,\big] \le 16^{t+\ell}M^{\lceil t/\ell\rceil} (32pk)^t + \delta. \label{eq:failure}
 \end{equation} 
\end{enumerate} 
\end{lemma} 

\begin{proof} 
Let $\calD$ be an $r$-wise independent distribution over $\{0,1\}^{Y_q}$; standard constructions~\cite{ABI86} show that $\calD$ can be sampled with $O(r \log |Y_q|) = O(r\log n)$ random bits. The marginal of $\calD$ on $\{0,1\}^{n\times q}$ naturally induces a distribution $\calR_\stars$ over subsets of $[n]$ via Definition~\ref{def:string-to-restriction}, where a draw $\bL \leftarrow \calR_\stars$ is defined to be $\rho(\bvrho,\bz)^{-1}(\ast)$ (i.e.~for all coordinates $i\in [n]$, $i \in \bL$ iff $\bvrho_{i,1} = \bvrho_{i,2} = \cdots =\bvrho_{i,q} = 1$).  Since $\calD$ is $r$-wise independent for $r \gg q$, we have that 
\[ \Prx_{\bL\leftarrow\calR_\stars}\big[i\in \bL\big]  = \Prx_{(\bvrho,\bz)\leftarrow \calD}\big[ \,\bvrho_{i,1} = \bvrho_{i,2} = \cdots =\bvrho_{i,q} = 1\,\big] =  2^{-q} = p, \] 
which establishes the second claim. The third claim follows by combining Theorem~\ref{thm:derand-H14}, Theorem~\ref{thm:HS16}, and Proposition~\ref{prop:swap}. 
\end{proof} 

\subsection{Pseudorandom simplification lemma for $\acz$ circuits} 

We will use the following instantiation of Lemma~\ref{lem:PSL} in our construction of a PRG for $\acz$ circuits: 

\begin{corollary}
\label{cor:multi-for-ac0} 
There is a universal constant $c > 0$ such that the following holds. Let $\mathscr{F} = (F_1,\ldots,F_M)$ be an ordered list of $Q$-clause $k$-CNFs with $\log M \geq k$ and $\eps_0 \in (0,1)$. There is a distribution $\calR_\stars$ over subsets of $[n]$ such that: 
\begin{enumerate}
\item A draw $\bL \leftarrow\calR_\stars$ can be sampled with $s = \log^c(MQ)\log(1/\eps_0)$ random bits. 
\item $\calR_\stars$ is $p$-regular for $p = \Omega(1/k)$.
\item A multi-switching lemma holds with respect to $\calR_\stars$: 
\[ \mathop{\Prx_{\bL\leftarrow\calR_\stars}}_{\brho\leftarrow\zo^{[n]\setminus \bL}}\big[\,\depth(\CCDT_{\log M}(\mathscr{F}\uhr\brho)) \ge \log(2M^{5}/\eps_0)\,\big] \le \eps_0. \]
\end{enumerate}  
\end{corollary}

\begin{proof}
Applying Lemma~\ref{lem:PSL} with $\ell = \log M$, we see that the failure probability (\ref{eq:failure}) can be bounded by 
\[ 
16^{t+\ell}M^{\lceil t/\ell\rceil} (32pk)^t + \delta \le 16^{t}M^4 
M^{\lceil t/\ell\rceil} 
(64)^{-t} + \delta < M^5 2^{-t} + \delta
\]
by choosing $p = \Omega(1/k)$.  We make this at most $\eps_0$ by choosing $t = \log(2M^{5}/\eps_0)$ and $\delta = \eps_0/2$. The bound on $s$ follows from the $d=3$ case of Theorem~\ref{thm:HS16} and our setting of parameters, and this completes the proof. 
\end{proof} 

Following the standard bottom-up approach to $\acz$ circuit lower bounds, we compose $d-1$ iterative applications of the pseudorandom multi-switching lemma of Corollary~\ref{cor:multi-for-ac0} to obtain our pseudorandom simplification lemma for $\acz$:

\begin{lemma}[Pseudorandom simplification lemma for $\acz$]
\label{lem:PSL-AC0} 
There is a universal constant $C > 0$ such that the following holds. Let $\calC$ be a size-$M$ depth-$d$ Boolean circuit over $\{0,1\}^n$ (so recall that $M \geq n$) and $\eps_1 \in (0,1)$. There is a distribution $\calR_\stars$ over subsets of $[n]$ such that 
\begin{enumerate} 
\item A draw $\bL\leftarrow\calR_\stars$ can be sampled with $s = O(2^{d} \log^{C}(M) \log(1/\eps_1))$ random bits.
\item $\calR_\stars$ is $p$-regular for $p = \Omega(1/\log^{d-1}(M))$. 
\item The following simplification lemma holds with respect to $\calR_\stars$:\[ \mathop{\Prx_{\bL\leftarrow\calR_\stars}}_{\brho\leftarrow\zo^{[n]\setminus \bL}}\big[\,\text{$\calC \uhr \brho$ is not a decision tree of depth $O(2^d\log(M/\eps_1))$}\,\big] \le \eps_1. \]
\end{enumerate} 
\end{lemma}

\begin{proof} 
Fix $t := \log(2dM^5/\eps_1)$.  

\medskip 
\noindent {\bf Preprocessing stage:} We begin with a zeroth stage of preprocessing to trim the bottom fan-in of $\calC$: applying Corollary~\ref{cor:multi-for-ac0} with $\mathscr{F}$ being the bottom layer gates of $\calC$ (viewed as depth-$2$ circuits of size $Q \le n$ and bottom fan-in $k=1$) and $\eps_0 = \eps_1/d$, we get that there is a distribution $\calR_\stars^{(0)}$ such that $\calR_\stars^{(0)}$ can be sampled with $s_0 := \log^{c}(Mn)\log(d/\eps_1)$ random bits (where $c$ is the universal constant from Corollary~\ref{cor:multi-for-ac0}), $\calR_\stars^{(0)}$ is $p_0$-regular for $p_0 = \Omega(1)$,  and 
\[ \mathop{\Prx_{\bL\leftarrow\calR_\stars^{(0)}}}_{\brho^{(0)}\leftarrow\zo^{[n]\setminus \bL}}\big[\,\calC \uhr \brho^{(0)} \text{~is not a $(t, \acz(\text{depth $d$, bottom fan-in $\log M$}))$-decision tree}\, \big] \le \frac{\eps_1}{d}. \] 
{\bf First stage:} Let $T^{(0)}$ be any good outcome of the zeroth stage above, a $(t, \acz($depth $d$, bottom fan-in $\log M))$-decision tree. Note that there are at most $2^t$ many $\acz(\text{depth $d$, bottom fan-in $\log M$})$ circuits at the leaves of this depth-$t$ decision tree $T^{(0)}$, each of size at most $M$. Fix any such circuit $\calC'$. Applying Corollary~\ref{cor:multi-for-ac0} to $\calC'$, with $\mathscr{F}$ being all its bottom layer depth-$2$ subcircuits of bottom fan-in $\log M$ (so $Q \le M$) and $\eps_0 = \eps_1/(d2^t)$, we get that there is a distribution $\calR_\stars^{(1)}$ such that $\calR_\stars^{(1)}$ can be sampled with $s_1 := \log^{c}(M^2) \log(d2^t/\eps_1)$  random bits, $\calR_\stars^{(1)}$ is $p_1$-regular for $p_1 = \Omega(1/\log M)$, and  
\[ \mathop{\Prx_{\bL\leftarrow\calR_\stars^{(1)}}}_{\brho^{(1)}\leftarrow\zo^{[n]\setminus \bL}}\big[\,\calC' \uhr \brho^{(1)} \text{~is not a $(2t, \acz(\text{depth $d-1$, bottom fan-in $\log M$}))$-decision tree}\, \big] \le \frac{\eps_1}{d\,2^t}. \] 
Taking a union bound over all the circuits at the leaves of $T^{(0)}$ (at most $2^t$ of them), we get that 
\[ \mathop{\Prx_{\bL\leftarrow\calR_\stars^{(1)}}}_{\brho^{(1)}\leftarrow\zo^{[n]\setminus \bL}}\big[\,T^{(0)} \uhr \brho^{(1)} \text{~is not a $(t+2t, \acz(\text{depth $d-1$, bottom fan-in $\log M$}))$-decision tree}\, \big] \le \frac{\eps_1}{d}. \] 
Let $T^{(1)}$ be any good outcome of the above, and consider any circuit $\calC''$ at a leaf of this depth-$3t$ decision tree. We note a subtlety at this point (this same subtlety is present in applications of the standard switching lemma): while $\calC''$ has at most $M$ gates in total from levels $1$ to $d-2$ (indeed, its number of gates in those layers is at most that of $\calC$), each of its bottom layer depth-$2$ subcircuits may have size as large as $M^2$. This is because the $M$-way AND of depth-$(\log M)$ decision trees, when expressed as depth-$2$ circuit, can have size as large as $M \cdot 2^{\log M} = M^2$. (And of course the same is true for the $M$-way OR.) Therefore from the second stage onwards, we will always apply Corollary~\ref{cor:multi-for-ac0} with $\mathscr{F}$ being a family of $M$ many $M^2$-clause $(\log M)$-CNFs (or DNFs), and so $Q = M^2$. \smallskip

\noindent {\bf The $i$-th stage:} We repeat for $d-2$ more stages, where in the $i$-th stage we consider a good outcome $T^{(i-1)}$ of the previous stage, a $((2^{i}-1)t, \acz(\text{depth $d-i+1$, bottom fan-in $\log M$}))$-decision tree.  Fix any subcircuit $\calC'''$ of at a leaf of this depth-$((2^i-1)t)$ decision tree $T^{(i-1)}$.  Applying Corollary~\ref{cor:multi-for-ac0} to $\calC'''$, with $\mathscr{F}$ being all its bottom layer depth-$2$ subcircuits of bottom fan-in $\log M$ (as noted above, we take $Q = M^2$) and 
\[ \eps_0 = \frac{\eps_1}{d\,2^{(2^{i}-1)t}}, \] 
we get that there is a distribution $\calD_{\stars}^{(i)}$ such that $\calR_\stars^{(i)}$ can be sampled with 
\[ s_i := \log^c(M^3) \log(1/\eps_0) = 2^i \cdot O(t\log^c(M)) \] 
random bits, $\calR_\stars^{(i)}$ is $p_i$-regular for $p_i = \Omega(1/\log M)$, and
\[ \mathop{\Prx_{\bL\leftarrow\calR_\stars^{(i)}}}_{\brho^{(i)}\leftarrow\zo^{[n]\setminus \bL}}\big[\,\calC''' \uhr \brho^{(i)} \text{~is not a $(2^i t, \acz(\text{depth $d-i$, bottom fan-in $\log M$}))$-decision tree}\, \big] \le \frac{\eps_1}{d\,2^{(2^{i}-1)t}}. \]  
(We have used the fact that $\log(2M^{5}/\eps_0) = (2^i-1)t + \log(2dM^{5}/\eps_1) = 2^it$.)  Taking a union bound over all the circuits at the leaves of $T^{(i-1)}$ (at most $2^{(2^i-1)t}$ of them), we get that 
\[ \mathop{\Prx_{\bL\leftarrow\calR_\stars^{(i)}}}_{\brho^{(i)}\leftarrow\zo^{[n]\setminus \bL}}\big[\,T^{(i-1)} \uhr \brho^{(i)} \text{~is not a $((2^{(i+1)}-1)t, \acz(\text{depth $d-i$, bottom fan-in $\log M$}))$-decision tree}\, \big] \le \frac{\eps_1}{d}. \] 
{\bf The overall distribution.} Composing all $d$ stages described above (including the zeroth preprocessing stage), we get an overall distribution $\calR_\stars$ where a draw $\bL \leftarrow \calR_\stars$ is simply 
\[ \bL = \bL^{(0)} \cap \bL^{(1)} \cap \cdots \cap \bL^{(d-1)}, \quad \bL^{(i)} \leftarrow \calR_\stars^{(i)} \text{~for all $0\le i \le d-1$}. \]
This distribution $\calR_\stars$ can be sampled with 
\[ \sum_{i=0}^{d-1} s_i = O( 2^{d}  \log^C(M)\log(1/\eps_1)) \]
random bits for some constant $C > 0$, $\calR_\stars$ is $p$-regular for 
\[ p = \prod_{i=0}^{d-1} p_i = \Omega(1/\log^{d-1}(M)), \] 
and by a union bound over the $d$ many failure probabilities of $\eps_1/d$ from each of the $d$ stages, we have that indeed 
\[ \mathop{\Prx_{\bL\leftarrow\calR_\stars}}_{\brho\leftarrow\zo^{[n]\setminus \bL}}\big[\,\text{$\calC \uhr \brho$ is not a depth-$((2^d-1)t)$ decision tree}\,\big] \le \eps_1. \]
Since $(2^d-1)t = O( 2^d \log(M/\eps_1))$ (using $d \leq M$ so $\log(2dM^5/\eps_1) \leq \log(2M^6/\eps_1)$), this completes the proof. 
\end{proof}

\subsection{Pseudorandom simplification lemma for sparse $\F_2$ polynomials} 

To motivate the parameter settings used in this subsection, we recall the discussion about multi-switching lemmas and sparse $\F_2$ polynomials right before Section \ref{sec:ccdt}; observe that both the $\ast$-probability $p$ and the degree of the $\F_2$ polynomials obtained below are independent of the failure probability $\eps_2$.
\begin{lemma}[Pseudorandom simplification lemma for sparse $\F_2$ polynomials] 
\label{lem:PSL-F2} 
There is a universal constant $C > 0$ such that the following holds. Let $P$ be an $S$-sparse $\F_2$ polynomial and $\eps_2 \in (0,1)$. There is a distribution $\calR_\stars$ over subsets of $[n]$ such that 
\begin{enumerate}
\item A draw $\bL \leftarrow \calR_\stars$ can be sampled with $s = \log^C(Sn) \log (1/\eps_2)$ random bits. 
\item $\calR_\stars$ is $p$-regular for $p = 2^{-O(\sqrt{\log S})}$. 
\item The following simplification lemma holds with respect to $\calR_\stars$: 
\[ \mathop{\Prx_{\bL\leftarrow\calR_\stars}}_{\brho\leftarrow\zo^{[n]\setminus \bL}}\big[\,\text{$P \uhr \brho$ is not a $\big(O(\sqrt{\log S}) + \log(2/\eps_2), \F_2(\text{degree $\sqrt{\log S}$})\big)$-decision tree}\,\big] \le \eps_2. \] 

\end{enumerate}
\end{lemma}

\begin{proof}
We observe that an $S$-sparse $\F_2$ polynomial is simply a $\PAR \circ \AND$ circuit with $S$ many bottom layer gates of unbounded fan-in.  With this point of view in mind, we apply Lemma~\ref{lem:PSL} with $\mathscr{F}$ being this family of $S$ many AND gates (viewed as depth-$2$ circuits of size $Q \le n$ and bottom fan-in $k=1$) and $\ell = \sqrt{\log S}$.  By choosing $t = A \cdot \sqrt{\log S} +  \log(2/\eps_2)$, $p = 2^{-B\sqrt{\log S}}$, and $\delta = \eps_2/2$, we get that the failure probability (\ref{eq:failure}) can be bounded by  
\begin{align*}  16^{t+\ell}S^{\lceil t/\ell\rceil} (32pk)^t + \delta &=
16^{(A+1)\sqrt{\log S} + \log(2/\eps_2)} \cdot S^{1+A} \cdot 2^{\sqrt{\log S} \cdot \log(2/\eps_2) } \cdot {\frac {32^{A \cdot \sqrt{\log S} + \log(2/\eps_2)}} {S^{AB} \cdot 2^{B \sqrt{\log S} \cdot \log(2/\eps_2)}}} + \frac{\eps_2}{2}\\
& < \eps_2,
\end{align*}
\ignore{We make this at most $\eps_2$ by choosing $t = 2\sqrt{\log S} + \log(2/\eps_2)$ and $\delta = \eps_2/2$.} where the inequality holds for a suitable choice of absolute constant values $A,B$.  The bound on $s$ follows from the $d=3$ case of Theorem~\ref{thm:HS16} and our setting of parameters,
and this completes the proof. 
\end{proof}

\section{PRGs for $\acz$ and sparse $\F_2$ polynomials from pseudorandom simplification lemmas} \label{sec:puttogether}

We will need the following easy fact for both our PRG constructions: we can derive from a $p$-regular distribution $\calR_\stars$ satisfying a pseudorandom simplification lemma (in the sense of our main results in the previous section, Lemmas~\ref{lem:PSL-AC0} and~\ref{lem:PSL-F2}) a distribution $\calR_\stars'$ supported entirely on sets of size $(pn)/2$, such that $\calR_\stars'$ also satisfies a pseudorandom simplification lemma with only a slightly worse failure probability.  More precisely, and in more generality:

\begin{proposition}[Condition on having sufficiently many stars]
\label{prop:condition} 
Fix any property $\Phi : \{0,1,\ast\}^n \to \zo$ of restrictions. 
Let $\calR_\stars$ be a $p$-regular distribution over subsets of $[n]$ and suppose  
\[ \mathop{\Prx_{\bL\leftarrow\calR_\stars}}_{\brho\leftarrow\zo^{[n]\setminus \bL}} \big[ \Phi(\brho) = 1 \big] \le \tau \]
for some $\tau > 0$.   Let $\calR_\stars'$ be the distribution of $\bL\leftarrow\calR_\stars$ conditioned on $\bL$ satisfying $|\bL| \ge (pn)/2$. Then 
\[ \mathop{\Prx_{\bL\leftarrow\calR_\stars'}}_{\brho\leftarrow\zo^{[n]\setminus \bL}} \big[ \Phi(\brho) = 1 \big] \le \frac{2\tau}{p}. \]
 \end{proposition} 

\begin{proof} 
Since $\calR_\stars$ is $p$-regular we have that $\Ex_{\bL\leftarrow\calR_\stars} \big[|\bL|\big] = pn$, and so
\[ \Prx_{\boldsymbol{L}\leftarrow\calR_\stars}\big[\boldsymbol{L} \in \supp(\calR_\stars')\big]  = \Prx_{\boldsymbol{L}\leftarrow\calR_\stars}\Big[|\boldsymbol{L}| \ge \frac{pn}{2}\Big] \ge \frac{p}{2}.\] 
Hence 
\begin{align*}
 \mathop{\Prx_{\bL\leftarrow\calR_\stars'}}_{\brho\leftarrow\zo^{[n]\setminus \bL}}\big[\Phi(\brho)=1\big]  
& =  \mathop{\Prx_{\bL\leftarrow\calR_\stars}}_{\brho\leftarrow\zo^{[n]\setminus \bL}}\big[\Phi(\brho)=1 \mid \bL \in\supp(\calR_\stars')\big]  \\
& \le  \mathop{\Prx_{\bL\leftarrow\calR_\stars}}_{\brho\leftarrow\zo^{[n]\setminus \bL}}\big[\Phi(\brho)=1\big] \cdot  \frac1{\Prx_{\bL\leftarrow\calR_\stars}[\bL \in \supp(\calR_\stars')]} \ \le \ \frac{2\tau}{p}. \qedhere
\end{align*}
\end{proof}

\subsection{PRGs for $\acz$ circuits}

\begin{reptheorem}{thm:first}
For every $d\ge 2$, $M \geq n$, and $\eps > 0$, there is an $\eps$-PRG for the class $\mathscr{C}$ of $n$-variable size-$M$ depth-$d$ circuits with seed length $\log^{d+O(1)}(M)\log(1/\eps)$. 
\end{reptheorem}

\begin{proof} 
Applying Proposition~\ref{prop:condition} to the pseudorandom simplification lemma of Lemma~\ref{lem:PSL-AC0}, we get that for all $\eps_1 > 0$, there is a distribution $\calR'_\stars$ over subsets of $[n]$ such that 
\begin{enumerate} 
\item A draw $\bL\leftarrow\calR'_\stars$ can be sampled with $s_\SL = O( 2^{d} \log^{C}(M) \log(1/\eps_1)$ random bits, where $C > 0$ is the universal constant from Lemma~\ref{lem:PSL-AC0}.\footnote{Recalling the definition of efficient samplability (Section~\ref{sec:prelim}), we note that assuming $p\ge 1/\poly(n)$ (as is the case here), if $\calR_\stars$ can be sampled efficiently with $s$ random bits then so can $\calR_\stars'$.}
\item Every $L \in \supp(\calR'_\stars)$ satisfies $|L| \ge pn$ where $p = \Omega(1/\log^{d-1}(M))$. 
\item For all $\calC \in \mathscr{C}$,  
\[ \mathop{\Prx_{\bL\leftarrow\calR'_\stars}}_{\brho\leftarrow\zo^{[n]\setminus \bL}}\big[\,\text{$\calC \uhr \brho$ is not a decision tree of depth $O(2^d\log(M/\eps_1))$}\,\big] \le \frac{\eps_1}{p}. \]
\end{enumerate}
Setting $\eps_1 = \eps p^2/(2\ln n)$ and taking $\mathscr{C}_\simple$ to be the class of depth-$t$ decision trees where 
\[ t = O(2^d \log(M/\eps_1)) = O(d \,2^d\log(M/\eps)), \]
we get that a draw $\bL \leftarrow\calR'_\stars$ can be sampled with  
\[ s_\SL = O(2^d \log^C(M) \log(1/\eps_1) )=  O(d 2^d \log^C(M) \log((\log M)/\eps)) \]
random bits, and $\calR'_\stars$ satisfies
 \[ \Ex_{\bL\leftarrow\calR'_{\stars}}\bigg[ \Prx_{\brho\leftarrow\zo^{[n]\setminus \bL}}\big[\, (\calC \uhr \brho) \notin \mathscr{C}_\simple\big] \bigg]  \le \frac{\eps p}{2\ln n} \] 
for all $\calC \in \mathscr{C}$. Since $\mathscr{C}_\simple$ is $0$-fooled by any $t$-wise independent distribution, we get from Theorem~\ref{thm:AW} that there is an $\eps$-PRG for $\mathscr{C}$ with seed length 
 \[ O(s_\SL + t \log n) \cdot p^{-1}\ln n  = \log^{d+O(1)}(M)\log(1/\eps),\]
 and this completes the proof. \ignore{\rnote{I apologize for how ridiculous this rnote is....the final claimed bound $ \log^{d+O(1)}(M)\log(1/\eps)$ is ignoring the $2^{d}$ type factor that's present in $s_\SL$ and in $t$ so if $d$ is superconstant and bigger than $\Theta(\log \log M)$(as it could be) it seems to me the final claimed bound is not strictly 100 percent accurate.  (Do you think of the ``For every $d \geq 2$'' in the theorem statement as meaning that $d$ is a constant?  If so, then this is a non-issue.)   Should we write ``$(O(1) \cdot \log M)^{^{d+O(1)}}\log(1/\eps)$'' instead?}}
\end{proof}

\subsection{PRGs for sparse $\F_2$ polynomials}

\begin{reptheorem}{thm:second}
For every $S = 2^{\omega(\log \log n)^2}$ and $\eps > 0$ there is an $\eps$-PRG for the class $\mathscr{C}$ of $n$-variable $S$-sparse $\F_2$ polynomials with seed length $2^{O(\sqrt{\log S})}\log(1/\eps)$. 
\end{reptheorem} 

\begin{proof} 
Applying Proposition~\ref{prop:condition} to the pseudorandom simplification lemma of Lemma~\ref{lem:PSL-F2}, we get that for all $\eps_2 > 0$, there is a distribution $\calR'_\stars$ over subsets of $[n]$ such that 
\begin{enumerate}
\item A draw $\bL \leftarrow \calR'_\stars$ can be sampled with $s_\SL = \log^C(Sn)\log(1/\eps_2)$ random bits, where $C > 0$ is the universal constant from Lemma~\ref{lem:PSL-F2}.
\item Every $L\in \supp(\calR'_\stars)$ satisfies $|L| \ge pn$ where $p = 2^{-O(\sqrt{\log S})}$. 
\item For all $P \in \mathscr{C}$, 
\[ \mathop{\Prx_{\bL\leftarrow\calR'_\stars}}_{\brho\leftarrow\zo^{[n]\setminus \bL}}\big[\,\text{$P \uhr \brho$ is not a $\big(O(\sqrt{\log S}) +  \log(2/\eps_2), \F_2(\text{degree $\sqrt{\log S}$})\big)$-decision tree}\,\big] \le \frac{\eps_2}{p}. \] 
\end{enumerate}
Setting $\eps_2 = \eps p^2/(2\ln n)$ and taking $\mathscr{C}_\simple$ to be the class of $(t,\F_2(\text{degree $\sqrt{\log S}$}))$-decision trees where 
\[ t = O(\sqrt{\log S}) + \log(2/\eps_2) = O(\sqrt{\log S})  +  \log(1/\eps) \]
(where the second equality uses $S = 2^{\omega(\log \log n)^2}$),
we get that a draw $\bL \leftarrow\calR'_\stars$ can be sampled with  
\[ s_\SL = \log^C(Sn) \log(1/\eps_2) =  O\big(\log^{C+ \frac1{2}}(Sn)\log(1/\eps)\big)  \]
random bits, and $\calR'_\stars$ satisfies
 \[ \Ex_{\bL\leftarrow\calR_{\stars}}\bigg[ \Prx_{\brho\leftarrow\zo^{[n]\setminus \bL}}\big[\, (P \uhr \brho) \notin \mathscr{C}_\simple\big] \bigg]  \le \frac{\eps p}{2\ln n}\] 
 for all $P \in \mathscr{C}$.

We claim that the class of $(t,\F_2(\text{degree $k$}))$-decision trees can be $\delta$-fooled with seed length 
\[ s_\PRG(\delta) =  k \cdot O(t + 2^k\log(1/\delta)) + O(t \log n); \]
we defer the proof of this claim to the next subsection (see Lemma~\ref{lem:fool-simple}).   Recalling our definition of $\mathscr{C}_\simple$ where $t = O(\sqrt{\log S}) + \log(1/\eps)$ and $k = \sqrt{\log S}$, it follows from this claim that $\mathscr{C}_\simple$ can be $(\eps p/(2\ln n))$-fooled with seed length 
\begin{align*}s_\PRG &= 2^{O(\sqrt{\log S})} \log(1/\eps)
+ O(t \log n) = 2^{O(\sqrt{\log S})} \log(1/\eps) + O(\sqrt{\log S} \log n) + \log(1/\eps) \log n\\
&=2^{O(\sqrt{\log S})} \log(1/\eps)
\end{align*}
(where we have again used $S = 2^{\omega(\log \log n)^2}$).
Now applying Theorem~\ref{thm:AW}, we get that there is an $\eps$-PRG for $\mathscr{C}$ with seed length 
\begin{align*}
(s_\SL + s_\PRG) \cdot p^{-1}\ln n &= \Big(O\big( \log^{C+\frac1{2}}(Sn) \log(1/\eps)\big) + 2^{O(\sqrt{\log S})}\log(1/\eps)\Big) \cdot 2^{O(\sqrt{\log S})} \ln n \\
&= 2^{O(\sqrt{\log S})} \log(1/\eps)
\end{align*} 
(where the last equality yet again uses $S = 2^{\omega(\log \log n)^2}$),
and this completes the proof. 
 \end{proof}

\subsubsection{Fooling depth-$t$ decision trees with degree-$k$ $\F_2$ polynomials at its leaves} 

We recall the following well-known result of Viola:
\begin{theorem}[\cite{Vio09}] 
\label{thm:viola}
The sum of $k$ independent $(\frac1{16}\, \delta^{2^{k-1}})$-biased distributions $\delta$-fools the class of degree-$k$ $\F_2$ polynomials. 
\end{theorem}

Earlier work of Lovett~\cite{Lov09} proved the weaker statement with $2^k$ independent copies instead of $k$. We note that Lovett's result suffices for our purposes. 

We will need a few simple facts about distributions. Recall that a distribution $\calD$ is a \emph{mixture} of
\emph{component} distributions $\calD^{(1)},\dots,\calD^{(\ell)}$ if there exist non-negative weights $w_1,\dots,w_\ell$ summing to 1 such that making a draw
from $\calD$ corresponds to first drawing $i \in [\ell]$ with probability $w_i$ and then making a draw from $\calD^{(i)}.$

\begin{fact} \label{fact:mix}
Let $\mathscr{C}$ be a class of functions and suppose that distributions $\calD^{(1)},\dots,\calD^{(\ell)}$ each $\delta$-fool $\mathscr{C}.$
Then any mixture $\calD$ of distributions $\calD^{(1)},\dots,\calD^{(\ell)}$ also $\delta$-fools $\mathscr{C}.$
\end{fact}

We say that a class $\mathscr{C}$ of Boolean functions is \emph{closed under reorientations} if for all $f \in \mathscr{C}$ and $y \in \zo^n$, the function $g(x) := f(x + y)$ is also in $\mathscr{C}$ (where addition is coordinate-wise over $\F_2$).
An easy consequence of Fact~\ref{fact:mix} is the following:

\begin{fact}
\label{fact:more-is-better} 
Let $\mathscr{C}$ be a class of functions, closed under reorientations, that is $\delta$-fooled by a distribution $\calD$. Let $\calD'$ be any other independent distribution. Then the distribution $\calD + \calD'$, where a draw from $\calD + \calD'$ is $\bx + \by$ where $\bx \leftarrow \calD$ and $\by\leftarrow\calD'$, also $\delta$-fools $\mathscr{C}$. 
\end{fact} 

Finally we recall the following which is an easy consequence of the definition of a $\delta$-biased distribution:
\begin{fact}[Conditioning a $\delta$-biased distribution]
\label{fact:conditioning}
Let $\calD$ be a $\delta$-biased distribution over $\zo^n$. Fix $i\in [n]$ and $b \in \zo$, and let $\calD'$ denote the distribution of $\bx\leftarrow\calD$ conditioned on $\bx_i = b$.  Then the marginal distribution of $\calD'$ on the coordinates in $[n] \setminus \{i\}$ is $2\delta/(1-\delta) \le 4\delta$ biased.\ignore{\lnote{The overall distribution $\calD'$ (over all of $\zo^n$) is \emph{not} $\delta$-biased since $\ds\Prx_{\bx \leftarrow\calD'}[x_i = b] = 1$.}}
\end{fact}


\begin{lemma}[Fooling decision trees with low-degree polynomials at leaves] 
\label{lem:fool-simple}
Let $\calD^{(1)}_{\text{$\delta'$-biased}},\ldots,\calD^{(k)}_{\text{$\delta'$-biased}}$ be $k$ independent $\delta'$-biased distribution where $\delta' = \frac1{16}\,\delta^{2^{k-1}} \cdot 4^{-t}$. Let $\calD_{\text{$t$-wise}}$ be an independent $t$-wise independent distribution.  Then the sum 
\[ \calD := \calD^{(1)}_{\text{$\delta'$-biased}} + \cdots + \calD^{(k)}_{\text{$\delta'$-biased}} + \calD_{\text{$t$-wise}} \]
$\delta$-fools the class of depth-$t$ decision trees with degree-$k$ polynomials at its leaves.  Since $\delta'$-biased distributions can be generated with seed length $O(\log n + \log(1/\delta'))$, and $t$-wise independent distributions with seed length $O(t \log n)$, we get that we can sample from $\calD$ using 
\[ k \cdot O(t + 2^k\log(1/\delta)) + O(t \log n) \] 
random bits. 
\end{lemma} 

The intuition underlying Lemma~\ref{lem:fool-simple} is as follows: 

\begin{enumerate}
\item $\calD_{\text{$t$-wise}}$ ensures that every branch of the decision tree is taken with the right probability. 
\item By Fact~\ref{fact:conditioning}, each $\calD_{\text{$\delta'$-biased}}$ remains $(\frac1{16}\,\delta^{2^{k-1}}4^{-t}) \cdot 4^t =  \frac1{16}\,\delta^{2^{k-1}}$-biased even when conditioned on a length-$t$ branch.  By Theorem~\ref{thm:viola}, their sum $\delta$-fools the degree-$k$ polynomial at the leaf. 
\end{enumerate}

\begin{proof} 
Let $F$ be computed by a depth-$t$ decision tree $T$ with degree-$k$ polynomials at its leaves.  We begin by noting that every branch $\pi$ of $T$ is taken with the right probability under a random draw from $\calD$: 
\begin{align*}
\Ex_{\by\leftarrow\calD}\big[ F(\by) = 1\big] &= \sum_{\pi \in T}\, \Prx_{\by\leftarrow\calD}[\text{~$\by$ follows $\pi$~}] \cdot \Ex_{\by\leftarrow\calD}\big[ (F\uhr\pi)(\by)\mid \text{$\by$ follows $\pi$} \big] \\
&= \sum_{\pi \in T}\, \Prx_{\bx\leftarrow\calU}[\text{~$\bx$ follows $\pi$~}] \cdot \Ex_{\by\leftarrow\calD}\big[ (F\uhr\pi)(\by)\mid \text{$\by$ follows $\pi$} \big], \tag*{(since $\calD$ is $t$-wise independent and $|\pi|\le t$)}
\end{align*}
so it remains to show that 
\[ \Ex_{\by\leftarrow\calD}\big[ (F\uhr\pi)(\by)\mid \text{$\by$ follows $\pi$} \big]  = \Ex_{\bx\leftarrow\calU}\big[ (F\uhr\pi)(\bx) \big] \pm \delta \qquad \text{for all $\pi \in T$.}  \]
Since for all $\pi\in T$ $F \uhr \pi$ is a degree-$k$ polynomial over the coordinates in $[n] \setminus \supp(\pi)$, it suffices to show that $\calD \uhr \pi$, the distribution of $\by\leftarrow\calD$ conditioned on $\by$ following $\pi$, $\delta$-fools the class of degree-$k$ polynomials over the coordinates in $[n] \setminus \supp(\pi)$.
 
Fix $\pi \in T$ and let $S$ denote $\supp(\pi)$. We will express $\calD \uhr {\pi}$ as a mixture of distributions, and argue that each component distribution in the mixture $\delta$-fools the class of degree-$k$ polynomials over the coordinates in $[n]\setminus S$. Recall that $\calD$ is the sum of $k+1$ many independent distributions
 \[ \calD = \calD^{(1)}_{\text{$\delta'$-biased}} + \cdots + \calD^{(k)}_{\text{$\delta'$-biased}} + \calD_{\text{$t$-wise}}, \]
and so a draw $y = z^{(1)} + \cdots + z^{(k+1)}$ is consistent with $\pi$ iff 
\[ z^{(1)}_S + \cdots + z^{(k+1)}_S = \pi_S. \] 
Therefore, $\calD \uhr \pi$ is a mixture of component distributions each of which is the sum of $k+1$ independent distributions.  Each component distribution is specified by a $(k+1)$-tuple $(\pi^{(1)},\dots,\pi^{(k+1)})$ 
where $\supp(\pi^{(i)}) = S$ for all $i \in [k+1]$ and 
\[ \bigoplus_{i\in [k+1]} \pi^{(i)}_S = \pi_S.\]
Given such a 
$(k+1)$-tuple $(\pi^{(1)},\dots,\pi^{(k+1)})$, the corresponding component distribution is
\begin{equation} \Big(\calD^{(1)}_{\text{$\delta'$-biased}} \uhr \pi^{(1)}\Big) + \cdots + \Big(\calD^{(k)}_{\text{$\delta'$-biased}} \uhr \pi^{(k)}\Big) + \Big(\calD_{\text{$t$-wise}} \uhr 
\pi^{(k+1)}\Big). \label{eq:mixand}
\end{equation}
(The values of the mixing weights for the components are\ignore{ easy to describe, but} irrelevant for our purposes.) By Fact~\ref{fact:conditioning}, the marginal distribution of each $\calD^{(i)}_{\text{$\delta$-biased}} \uhr \pi^{(i)}$ on the coordinates in $[n]\setminus S$ is 
\[ \delta' \cdot 4^{|\pi^{(i)}|} \le  \left(\frac1{16}\,\delta^{2^{k-1}}4^{-t}\right) \cdot 4^t =  \frac1{16}\,\delta^{2^{k-1}} \]
biased, and hence by Viola's theorem (Theorem~\ref{thm:viola}) their sum 
\[  \Big(\calD^{(1)}_{\text{$\delta'$-biased}} \uhr \pi^{(1)}\Big) + \cdots + \Big(\calD^{(k)}_{\text{$\delta'$-biased}} \uhr \pi^{(k)}\Big)  \] 
$\delta$-fools the class of degree-$k$ polynomials over the coordinates in $[n] \setminus S$. By Fact~\ref{fact:more-is-better}, so does the distribution in (\ref{eq:mixand}).  By Fact~\ref{fact:mix} the mixture distribution $\calD \uhr \pi$ likewise
$\delta$-fools the class of degree-$k$ polynomials over the coordinates in $[n] \setminus S$, and the proof is complete.
\end{proof} 

\section*{Acknowledgements} 

We thank Prahladh Harsha and Srikanth Srinivasan for helpful discussions.

\bibliography{allrefs}{}
\bibliographystyle{alpha}

\appendix

\section{Proof sketch of Theorem~\ref{thm:HSL14-canonical}} \label{sec:HSL14-proof}

We sketch a proof of the following:

\begin{theorem}
\label{thm:our-HSL14}
Let $\mathscr{F} = (F_1,\ldots,F_M)$ be an ordered collection of $k$-CNFs. Then for all $t ,\ell \in \N$,
\[ \Prx_{\brho\leftarrow\calR_p}[\, \depth(\CCDT_\ell(\mathscr{F}\uhr\brho)\ge t\,] \le M^{\lceil t/\ell\rceil} (32pk)^t.  \] 
\end{theorem} 

Our proof sketch of Theorem~\ref{thm:our-HSL14} is carried out in the ``encoding-decoding'' framework of Razborov's alternative proof~\cite{Raz95} of the H{\aa}stad's original switching lemma~\cite{Hastad86}, Theorem \ref{thm:HSL}.  (For a detailed exposition of Razborov's proof technique see~\cite{Beame:94,Tha09} and Chapter \S14 of~\cite{AB09}.)
  We emphasize that the ideas in our proof of Theorem~\ref{thm:our-HSL14} are all from \cite{Has14}, but in our view the encoding--decoding presentation is more amenable to the derandomization that we ultimately require than the conditioning-based inductive argument given in \cite{Has14}.  We also note that a similar proof based on the encoding--decoding framework appears in Section 7 of~\cite{Tal:15tightbounds}.

\subsection{Bad restrictions and the structure of witnessing paths} 

Fix $\mathscr{F} = (F_1,\ldots,F_M)$ and consider the set $\calB \sse \{0,1,\ast\}^n$ of all \emph{bad} restrictions $\rho$, namely the ones such that 
\[ \depth(\CCDT_{\ell}(\mathscr{F}\uhr\rho)) \ge t. \] 
Fix any bad restriction $\rho \in \calB$. Recalling our definition of the set of canonical common partial decision trees (Definition~\ref{def:ccpdt}), there exists a canonical common $\ell$-partial decision tree $T \in \CCDT_\ell(\mathscr{F} \uhr \rho)$ and a path $\Pi$ of length exactly $t$ through $T$. Furthermore, we have that 
\begin{enumerate} 
\item There exist indices $1\le i_1 \le i_2 \le \cdots \le i_u \le M$ where $u \le \lceil t/\ell\rceil$, and 
\item $\Pi = \pi^{(1)} \circ \cdots \circ \pi^{(u)}$, where for all $j\in [u]$, we have that $\supp(\pi^{(j)}) =  \supp(\eta^{(j)})$ where  $\eta^{(j)}$ is a path through the canonical decision tree 
\[ \CDT(F_{i_j} \uhr \rho \circ \pi^{(1)}\circ\cdots\circ \pi^{(j-1)}).\]  
Furthermore, for every $j\in [u-1]$ we have that $\eta^{(j)}$ is a full path of length between $\ell+1$ and $\ell+k$ through the CDT, and $\eta^{(u)}$ is a path of length exactly $t - \sum_{j=1}^{u-1} |\supp(\eta^{(j)})|$. (Note that $\eta^{(u)}$ is not necessarily a full path.)

\end{enumerate}  
(Note that by (2), these restrictions $\pi^{(j)}$ are supported on mutually disjoint sets of coordinates.) 

\subsection{Encoding bad restrictions $\rho$} 

Recalling the statement of Theorem~\ref{thm:our-HSL14}, our goal is to bound $\Prx_{\brho\leftarrow\calR_p}[\, \text{$\brho \in \calB$}\,]$, the weight of the set $\calB$ of bad restrictions under $\calR_p$. To do so, we define an encoding of each bad restriction $\rho \in \calB$ as a different restriction $\rho' \in \{0,1,\ast\}^n$ and a small amount (say at most $m$ bits) of ``auxiliary information":  
\begin{align*}
\encode &: \calB \to \{0,1,\ast\}^n \times \zo^m \\
\encode(\rho) &= (\rho', \text{auxiliary information}) 
\end{align*}
This encoding should satisfy two key properties. First, it should be uniquely decodable, meaning that one is always able to recover $\rho$ given $\rho'$ and the auxiliary information; equivalently, the function $\encode(\cdot)$ is an injection.  Second, $\rho'$ should extend $\rho$ by exactly $t$ bits, meaning that $\supp(\rho) \sse \supp(\rho')$ and $|\supp(\rho') \setminus \supp(\rho)| = t$.  From this second property we get that 
\[ \frac{\Prx_{\brho\leftarrow\calR_p}[\, \brho = \rho'\,]}{\Prx_{\brho\leftarrow\calR_p}[\,\brho = \rho\,]} = \left(\frac{1-p}{2p}\right)^t, \] 
i.e.~that the weight of $\rho'$ under $\calR_p$ is larger than that of $\rho$ by a $O(p)^{-t}$ multiplicative factor.  It is not hard to see that together, these two properties imply that total weight of all bad restrictions with the \emph{same} auxiliary information is at most $O(p)^t$.  To complete the proof of Theorem~\ref{thm:our-HSL14}, we then bound the overall weight of $\calB$ via a union bound over all $2^m$ possible strings of auxiliary information, giving us a failure probability of 
\begin{equation} 2^m \cdot \left(\frac{2p}{1-p}\right)^t. 
\label{eq:failure-probability}
\end{equation} 

We now describe the encoding in more detail.  Given a bad restriction $\rho \in \calB$, the extension $\rho'$ of $\rho$ will be 
\begin{equation} \rho' = \rho\circ \sigma^{(1)} \circ\cdots\circ\sigma^{(u)}, \qquad u \le \lceil t/\ell \rceil \label{eq:our-encoding}
\end{equation} 
where $\sigma^{(j)}$ is a restriction that is supported on the same coordinates as $\pi^{(j)}$ for all $j\in [u]$.  (Hence these restrictions $\sigma^{(j)}$'s are supported on mutually disjoint sets of coordinates, every $\sigma^{(j)}$ has length between $\ell + 1$ and $\ell+k$, except $\sigma^{(u)}$ which has length $t - \sum_{j=1}^{u-1} |\supp(\sigma^{(j)})|$ and is not necessarily a full path.) 

We now define these restrictions $\sigma^{(j)}$. Recall that $\supp(\pi^{(j)}) = \supp(\eta^{(j)})$ where $\eta^{(j)}$ is a full path of length between $\ell+1$ and $\ell+k$ through the canonical decision tree $\CDT(F_{i_j} \uhr \rho \circ \pi^{(1)}\circ\cdots\circ \pi^{(j-1)})$, a full path witnessing that fact that $\CDT(F_{i_j} \uhr \rho \circ \pi^{(1)}\circ\cdots\circ \pi^{(j-1)}) > \ell$. That is, $\eta^{(j)}$ is a full path witnessing the fact that $\rho \circ \pi^{(1)}\circ\cdots\circ \pi^{(j-1)}$ is a bad restriction for the usual switching lemma, and $\pi^{(j)}$ is an assignment to the variables in $\supp(\eta^{(j)})$. (Once again this is with the possible exception of the segment $\pi^{(u)}$ of $\Pi$, which has length $t - \sum_{j=1}^{u-1} |\supp(\pi^{(j)})|$ and is not necessarily a full path.) Razborov's encoding--decoding proof of the usual switching lemma defines an encoding of this bad restriction $\rho \circ \pi^{(1)}\circ\cdots\circ \pi^{(j-1)}$ to an extension 
\[ \rho^{(j)} := \rho\circ \pi^{(1)} \circ \cdots \circ \pi^{(j-1)} \circ \sigma^{(j)} \] 
where $\sigma^{(j)}$ is supported on the same $\ell$ coordinates as $\eta^{(j)}$ (and hence $\pi^{(j)}$ as well). Razborov's proof hinges on the fact that given the $k$-CNF $F$, this encoding $\rho^{(j)}$, and a small amount of auxiliary information, one is able  to recover the bad restriction $\rho\circ \pi^{(1)}\circ\cdots\circ \pi^{(j-1)}$; that is, one as able to ``undo" $\sigma^{(j)}$ in $\rho^{(j)}$, flipping the coordinates in $\supp(\sigma^{(j)})$ from $\zo$ back to $\ast$.  This restriction $\sigma^{(j)}$ as defined in Razborov's proof is precisely the $\sigma^{(j)}$ we will use in our encoding (\ref{eq:our-encoding}).

We summarize the discussion above in the following fact: 

\begin{fact}[Main lemma in encoding--decoding proof of the usual switching lemma, notation specialized to our current context]
\label{fact:raz} 
Let $F_{i_j}$ be a $k$-CNF, $\rho \circ \pi^{(1)}\circ\cdots \pi^{(j-1)}$ be a restriction, and $\eta^{(j)}$ be  a path in $\CDT(F_{i_j}\uhr \rho\circ \pi^{(1)} \circ \cdots \circ \pi^{(j-1)})$.  There is a restriction $\sigma^{(j)}$ to the coordinates in $\supp(\eta^{(j)})$ such that given 
\begin{enumerate}
\item The $k$-CNF $F_{i_j}$, 
\item The restriction $\rho^{(j)} = \rho \circ \pi^{(1)} \circ \cdots \circ \pi^{(j-1)} \circ \sigma^{(j)}$, 
\item $|\supp(\eta^{(j)})|\cdot (2 + \log k)$ bits of auxiliary information $\iota(\rho \circ \pi^{(1)}\circ\cdots \circ\pi^{(j-1)}, F_{i_j})$,
\end{enumerate} 
a decoder is able to recover the restriction $\pi^{(1)} \circ \cdots \circ \pi^{(j-1)}$. 

Furthermore, if $\eta^{(j)}$ is a \emph{full} path in $\CDT(F_{i_j}\uhr \rho\circ \pi^{(1)} \circ \cdots \circ \pi^{(j-1)})$ (recall the definition of a full path given in Definition~\ref{def:full}) then given 
\begin{enumerate}
\item The $k$-CNF $F_{i_j}$, 
\item \emph{Any extension $\varrho^{(j)}$ of} the restriction $\rho^{(j)} = \rho \circ \pi^{(1)} \circ \cdots \circ \pi^{(j-1)} \circ \sigma^{(j)}$, 
\item $|\supp(\eta^{(j)})|\cdot (2 + \log k)$ bits of auxiliary information $\iota(\rho \circ \pi^{(1)}\circ\cdots \circ\pi^{(j-1)}, F_{i_j})$,\end{enumerate}
a decoder is able to ``undo" $\sigma^{(j)}$ in $\varrho^{(j)}$, by which we mean that she is able to recover the restriction $\bar{\varrho}^{(j)}$ where 
\[ 
\bar{\varrho}^{(j)}_i = \begin{cases} 
\ast & \text{if $i\in \supp(\sigma^{(j)})$} \\
\varrho^{(j)}_i & \text{otherwise.} 
\end{cases} 
\] 
\end{fact}

\subsection{Our auxiliary information}  
\label{sec:aux}
We will provide the decoder with 
\begin{enumerate}
\item $u \log M$ bits of information specifying the $u$ indices $i_1,\ldots,i_u \in [M]$.
\item  The auxiliary information $\iota(\rho \circ \pi^{(1)}\circ\cdots \circ\pi^{(j-1)}, F_{i_j})$ for all $j\in [u]$ (as defined in Fact~\ref{fact:raz}), a total of 
\[ \sum_{j=1}^u |\supp(\eta^{(j)})| \cdot (2+ \log k) = t \cdot (2 + \log k) \] 
many bits.
\item $t$ bits of information specifying the length-$t$ path $\Pi = \pi^{(1)} \circ \cdots \circ \pi^{(u)}$ through $\CCDT_\ell(\mathscr{F} \uhr \rho)$. 
\end{enumerate} 
This is a total of 
\[ m := u\log M + t\log k + 3t \] 
bits of auxiliary information; recalling equation (\ref{eq:failure-probability}) and the preceding discussion, to establish Theorem~\ref{thm:our-HSL14} it remains to argue that the map $\encode(\rho) = (\rho',\text{auxiliary information})$ is indeed invertible. 

\subsection{Decoding} 

Fix $\mathscr{F} = (F_1,\ldots,F_M)$ and consider a bad restriction $\rho \in\calB$, one such that 
\[ \depth(\CCDT_{\ell}(\mathscr{F}\uhr\rho)) \ge t. \] 
Let $\Pi = \pi^{(1)} \circ \cdots \circ \pi^{(u)}$ be a path of length $t$ through a canonical common $\ell$-partial decision tree $T \in \CCDT_\ell(\mathscr{F} \uhr \rho)$ that witnesses the badness of $\rho$.  We claim that for all $j\in [u]$, given 
\begin{enumerate}
\item The family of $k$-CNFs $\mathscr{F}$, 
\item The ``hybrid" restriction $\varrho^{(j)} := \rho \circ \pi^{(1)}\circ\cdots\circ \pi^{(j-1)} \circ \sigma^{(j)} \circ\cdots \circ \sigma^{(u)}$, 
\item The auxiliary information described in Section~\ref{sec:aux}, 
\end{enumerate} 
the decoder can recover the ``next" hybrid restriction $\varrho^{(j+1)} := \rho\circ \pi^{(1)} \circ \cdots \pi^{(j)} \circ \sigma^{(j+1)}\circ \cdots \circ \sigma^{(u)}$.  Before justifying this claim, we note that from this claim we get that the map $\encode(\rho) = (\rho',\text{auxiliary information})$ is indeed invertible, i.e.~that given $\rho'$ as defined in (\ref{eq:our-encoding}) and the auxiliary information described above, we can recover $\rho$ (this would complete our proof of Theorem~\ref{thm:our-HSL14}). To see this, we first observe that $\rho'$ is simply $\varrho^{(1)}$. Applying the claim $u$ times the decoder is able to iteratively recover $\varrho^{(2)},\ldots,\varrho^{(u+1)} = \rho\circ\pi^{(1)}\circ\cdots\circ \pi^{(u)}$, and having done so she will have identified $\supp(\pi^{(1)}\circ\cdots\circ \pi^{(u)})$. With this information she is then able to recover $\rho$ from $\rho^{(u+1)}$ (simply by flipping the bits in $\supp(\pi^{(1)}\circ\cdots\circ \pi^{(u)})$ back to $\ast$'s).

We now show how the decoder obtains $\varrho^{(j+1)}$ from $\varrho^{(j)}$ for all $j \in [u]$. First, since the auxiliary information specifies $i_j \in [M]$ she is able to identify $F_{i_j}$ within $\mathscr{F}$. Next, 
\begin{itemize}
\item for $j\in [u-1]$, we recall that $\eta^{(j)}$ is a full path in $\CDT(F_{i_j} \uhr \rho\circ \pi^{(1)} \circ \cdots \circ \pi^{(j-1)})$ and hence we may apply the ``Furthermore" part of Fact~\ref{fact:raz} to ``undo" $\sigma^{(j)}$ in $\varrho^{(j)}$ and obtain the restriction $\rho \circ \pi^{(1)} \circ \cdots \pi^{(j-1)} \circ \sigma^{(j+1)}\circ \cdots \circ \sigma^{(u)}$; 
\item for $j=u$, while $\eta^{(u)}$ is not necessarily a full path in $\CDT(F_{i_u} \uhr \rho\circ \pi^{(1)}\circ \cdots \circ \pi^{(u-1)})$ we observe that $\varrho^{(u)}$ is simply $\rho^{(u)}$, and hence we may apply the first part of Fact~\ref{fact:raz} to obtain the restriction $\rho\circ\pi^{(1)}\circ\cdots \circ \pi^{(u-1)}$. 
\end{itemize}
In either case, since our auxiliary information to the decoder specifies the values of $\pi^{(j)}$ on $\supp(\sigma^{(j)})$, the decoder is able to fill in these coordinates accordingly to obtain $\varrho^{(j+1)}$. 
\end{document}